\documentclass{LMCS}

\usepackage{amsmath,amssymb,amsbsy,latexsym,qsymbols,stmaryrd,color}
\usepackage{mathpartir}
\usepackage{enumerate,hyperref}

\theoremstyle{plain}\newtheorem{claim}[thm]{Claim}

\newcommand{\reptext}{\ensuremath{\mathsf{Rep}}}
\newcommand{\replone}[1]{\ensuremath{\reptext_{#1}}}
\newcommand{\repl}[2]{\replone{#1}\fillhole{#2}}

\newcommand{\capa}{\ensuremath{\mathsf{cap}}}

\newcommand{\rcap}{\arr{\capa}}

\newcommand{\Rcapar}[1]{\Ar{\langle\sf{#1}\rangle}}
\newcommand{\Rcap}{\Rcapar{\capa}}

\newcommand{\ds}{\mathsf{sd}}
\newcommand{\dd}{\mathsf{dd}}
\newcommand{\C}{\mbox{$\mathcal{C}$}}

\newcommand{\F}{\mathcal{F}}

\newcommand{\Pbb}{\MAIFsyn}

\newcommand{\mactext}[1]{\mathtt{#1}}
\newcommand{\cred}{\leadsto}
\newcommand{\Cred}{\leadsto^+}

\newcommand{\frozen}[1]{\ensuremath{\mathsf{froz}_N(#1)}}

\newcommand{\ttrue}{\mbox{\textsl{tt}}}
\newcommand{\ffalse}{\mbox{\textsl{ff}}}
\newcommand{\transTuring}{\mbox{$\rightarrowtail\hspace{-.3cm}\rightarrowtail$}}

\newcommand{\true}{\mbox{\textsf{true}}}

\newcommand{\nil}{\mbox{$\mathbf{0}$}}

\newcommand{\fn}[1]{\ensuremath{\mathrm{fn}(#1)}}

\newcommand{\hole}{\mbox{$[\,]$}}

\def\eqdef{\stackrel{\mathrm{{{def}}}}{=}}
\newcommand{\Rar}{\mbox{$\Longrightarrow$}}
\newcommand{\rar}{\mbox{$\longrightarrow$}}
\newcommand{\RR}{\ensuremath{\mathcal{R}}}

\newcommand{\inamb}[1]{\ensuremath{\mathsf{in~}#1}}
\newcommand{\outamb}[1]{\ensuremath{\mathsf{out~}#1}}
\newcommand{\openamb}[1]{\ensuremath{\mathsf{open~}#1}}
\newcommand{\amb}[2]{\ensuremath{#1[#2]}}
\newcommand{\mess}[1]{\ensuremath{\langle#1\rangle}}

\newcommand{\sat}{\mbox{$\models$}}
\newcommand{\eqL}{\mbox{$=_L$}}

\newcommand{\AAA}{\mbox{$\mathcal{A}$}}
\newcommand{\A}{\mbox{$\mathcal{A}$}}
\newcommand{\BBB}{\mbox{$\mathcal{B}$}}
\newcommand{\B}{\mbox{$\mathcal{B}$}}

\newcommand{\ltrue}{\mbox{\textbf{$\top$}}}

\newcommand{\rtr}{\ensuremath{\triangleright}}

\newcommand{\sometime}{\mbox{$\Diamond$}}

\newcommand{\reveal}{\mbox{$\circledR$}}
\newcommand{\hide}{\mbox{$\oslash$}}

\newcommand{\openhole}{\{\!|}
\newcommand{\closehole}{|\!\}}
\newcommand{\fillhole}[1]{\ensuremath{\openhole#1\closehole}}
\renewcommand{\hole}{\fillhole{\,}}
\newcommand{\QQ}{\mbox{${\mathcal Q}$}}
\newcommand{\length}[1]{\ensuremath{\mathsf{length}(#1)}}
\newcommand{\bigind}{\sim_{\rm{ind}}}

\newcommand{\thereis}[2]{\exists \; #1 \, . \, #2} 

\def\defiDS{\stackrel{\mathrm{{{def}}}}{=}}

\newcommand{\satDS}[2]{#1 \models #2 }

\newcommand{\orr}{ \vee } 
 
\newcommand{\zero}{{\boldmath 0}} 
\newcommand{\all}[2]{\forall \; #1 \, . \, #2} 

\newcommand{\at}[2]{ #1 @ #2} 

\newcommand{\limp}{\triangleright}  

\newcommand{\myneg}{\neg\,} 

\newcommand{\tkpPOPL}{3pt}  

\def\trans#1{\mbox{{\sc{\tt #1}}}}   

\newcommand{\Mybar}
{
{
\vskip .3cm 
 \rule{\hsize}{0.3mm} \vskip .7cm 
}
}

\newcommand{\shortaxiomC}[2]{
 \displaystyle{ \over #1} \; \trans{#2} }

\newcommand{\shortinfruleN}[3]{  
\displaystyle{#1 \over #2} {\; \trans{#3}}}

\def\sub#1#2{\{\raisebox{.5ex}{\small$#1$}\! / \!\mbox{\small$#2$}\}}
\newcommand{\andalso}{\quad\quad}  

\def\abs#1#2{(#1)\: #2}                     

\newcommand{\msg}[1]{\langle{#1}\rangle } 

\newcommand{\ina}[1]{{\mathsf{in}} \, #1}
\newcommand{\open}[1]{{\mathsf{open}} \, #1}
\newcommand{\tkp}{2pt}  

\newcommand{\mboxB}[1]{\mbox{\it (#1)}}
\def\midd{\; \; \mbox{{$\mid$}}\;\;} 
\newcommand{\equivE}{\mathrel{\equiv_{\rm E}}}  

\newcommand{\MA}{\mbox{MA}}
\newcommand{\MAIF}{\mbox{MA$_{\rm{IF}}$}}
\newcommand{\MAIFsyn}{\mbox{MA$^{\rm{s}}_{\rm{IF}}$}}

\newcommand{\IF}{image-finiteness}

\def\arr#1{\:\stackrel{#1}{\mbox{\rightarrowfill}}\:}

\newcommand{\Ar}[1]{\:\stackrel{{#1}}{\Longrightarrow} \:}

\newcommand{\Stat}[2]{{"={^{(#1 ,  #2)^\star}}!>"}}  
\newcommand{\StatOI}[1]{\Stat{\outamb{n}}{\inamb n}}
\newcommand{\StatIO}[1]{\Stat{\inamb n}{\outamb n}}

\def\reff#1{(\ref{#1})}       

\newcommand{\bisMOD}{\mathrel{\simeq_{\mathrm{int}}}} 
\newcommand{\intbis}{\ensuremath{\bisMOD}}

\renewcommand{\tilde}{\widetilde}

\mathcode`\!="4021       
\mathcode`\|="326A       
\mathcode`\.="602E       

\renewcommand{\true}{\ltrue}

\newcommand{\N}{\mathbb{N}}

\newcommand{\R}{\mathcal{R}}

\newcommand{\non}{\lnot}

\newcommand{\may}[1]{\langle\langle #1 \rangle\rangle}
\newcommand{\must}[1]{[\![ #1 ]\!]}
\newcommand{\Fmeta}[1]{\may{#1}}
\newcommand{\Ftame}[1]{\must{#1}}

\newcommand{\etarew}{\ensuremath{\longrightarrow_\eta}}

\usepackage{graphics} 

\newcommand\new[0]{\reflectbox{\ensuremath{\mathsf{N}}}} 
\newcommand\jamie\new

\newif\ifpopl \poplfalse   
\newif\iflong \longtrue   
\input{commonmacro.sty}
\input{davide.sty}
\input{francesca.sty}
\input{newmacro.sty}

\newcommand{\MAIFss}{MA$_{{\rm IF}}^{{\rm s,s}}$}

\newcommand{\etarewH}{\ensuremath{\longrightarrow_{\eta\mathtt{h}}}}

\def\doi{4 (3:4) 2008}
\lmcsheading%
{\doi}
{1--44}
{}
{}
{Nov.~\phantom{0}8, 2007}
{Sep.~\phantom{0}4, 2008}
{}   

\begin{document}

\title{Separability in the Ambient Logic\rsuper*
}

\author[D.~Hirschkof]{Daniel Hirschkoff\rsuper a}
\address{{\lsuper a}ENS Lyon, Universit\'e de Lyon, CNRS, INRIA -- France}
\email{Daniel.Hirschkoff@ens-lyon.fr}
\thanks{{\lsuper a}Work supported by the french projects ACI GEOCAL and ANR CHoCo.}
\author[{\'E}.~ Lozes]{{\'E}tienne Lozes\rsuper b}
\address{{\lsuper b}LSV, ENS Cachan, CNRS -- France}
\email{lozes@lsv.ens-cachan.fr}
\author[D.~Sangiorgi]{Davide Sangiorgi\rsuper c}
\address{{\lsuper c}Universit{\`a} di Bologna -- Italy}
\email{Davide.Sangiorgi@cs.unibo.it}
\thanks{{\lsuper c}Work supported by european project Sensoria, italian MIUR
  Project n.\ 2005015785, "Logical Foundations of Distributed Systems
  and Mobile Code".}

\keywords{Process algebra, modal logic, Mobile Ambients, spatial logic}
\subjclass{F.3.2, F.4.1}

\titlecomment{{\lsuper*}This work is a revised and extended version of parts
of~\cite{Sangiorgi::ExtInt::01} and~\cite{Sedal} (precisely, those
parts that deal with issues related to separability).
}

\begin{abstract}
\noindent   The \emph{Ambient Logic} (AL) has been proposed for
  expressing properties of process mobility in the calculus of Mobile
  Ambients (MA), and as a basis for query languages on semistructured
  data.

  We study some basic questions concerning the discriminating power of
  AL, focusing on the equivalence on processes induced by the logic
  ($\eqL$).  As underlying calculi besides MA we consider a
  subcalculus in which an image-finiteness condition holds and that we
  prove to be Turing complete.  Synchronous variants of these calculi
  are studied as well.

  In these calculi, we provide two operational characterisations of
  $\eqL$: a coinductive one (as a form of bisimilarity) and an
  inductive one (based on structual properties of processes).  After
  showing $\eqL$ to be stricly finer than barbed congruence, we
  establish axiomatisations of $\eqL$ on the subcalculus of MA (both
  the asynchronous and the synchronous version), enabling us to relate
  $\eqL$ to structural congruence.  We also present some
  (un)decidability results that are related to the above separation
  properties for AL: the undecidability of $\eqL$ on MA and its
  decidability on the subcalculus.
\end{abstract}

\maketitle

\section{Introduction}




This paper is devoted to the study of the \emph{Ambient
Logic}~\cite{CaGo00L} (AL), a modal logic for expressing properties of
Mobile Ambients~\cite{CaGo98} (MA) processes. The model of Mobile
Ambients is based on the notion of locality (an ambient is a named
locality), and interaction in MA appears as movement of localities.
Localities may be nested, as in $a[ P ~|~ b[Q]~|~ c[R]]$, which
describes an ambient $a$ containing a process $P$ as well as two
sublocalities named $b$ and $c$. 
%

An ambient can be thought of as a labelled tree. The sibling relation
on subtrees represents spatial contiguity; the subtree relation
represents spatial nesting.  A label may represent an ambient name or
a capability; moreover, a replication tag on labels indicates the
resources that are persistent.
The trees are unordered: the order of the children
of a node is not important.  
As an example, the process $P \defiDS !\amb a {\inamb c} | \openamb a
. \amb b \nil$ can be thought of as a tree with $\openamb a . \amb b
\nil$ on the roots node and $\inamb c$ on a child node labeled with a.
The  replication $!a$ indicates that the resource 
$ \amb a {\inamb c} $ is persistent: unboundedly many such ambients 
 can be spawned. By contrast, $\openamb a$ is
ephemeral: it can open only one ambient.

Syntactically, each tree is finite. Semantically, however, due to
replications, a tree is an infinite object.  As a consequence, the
temporal developments of a tree can be quite rich.  The process $P$
above (we freely switch between processes and their tree
representation) has only one reduction, to $\inamb c | ! \amb a
{\inamb c} | \amb b \nil$.  However, the process $! \amb a {\inamb c}
|! \openamb a . \amb b \nil$ can evolve into any process of the form
\begin{mathpar}
  { \inamb c | \ldots | \inamb c | \amb b \nil | \ldots | \amb b \nil
    | ! \amb a {\inamb c} | ! \openamb a . \amb b \nil\,.}
\end{mathpar}
In general, a tree may have an infinite temporal branching, that is,
it can evolve into an infinite number of trees, possibly quite
different from each other (for instance, pairwise behaviourally
unrelated).  Technically, this means that the trees are not
\emph{image-finite}, where image-finite indicates a finiteness on the
temporal branching of a process (we will come back to the definition
of \IF\ later).
 
Although the MA calculus often includes name restriction, $(\nu n)P$,
reminiscent of the pi-calculus, we will omit this construction (unless
we mention it explicitly), and will refer to public MA, or simply MA,
for the calculus without name restriction.

 In summary, \MA\ is a calculus of dynamically-evolving unordered
 edge-labelled trees. 
 %
   AL is a logic for reasoning on
 such trees.  The actual definition of satisfaction of the
 formulas is given on \MA\ processes quotiented by a relation of
 \emph{structural congruence},  $\equiv$, 
which equates processes with
 the same tree representation. (This relation is similar to
 Milner's structural congruence for the $\pi$-calculus~\cite{Mil99}.)
 
 AL has also been advocated as a foundation of query
 languages for semistructured
 data~\cite{CardelliSemiStructuredData}. 
 Here, the laws of the logic are used to describe query rewriting
 rules and query optimisations.  This line of work exploits the
 similarities between dynamically-evolving edge-labelled trees, 
 underlying the ambient computational model,  
 and standard models of semistructured data.

 AL has a   connective that talks about
\emph{time}, that is, how   processes  can evolve.
The  formula $\diamond \AAA$ is satisfied by those processes 
with a   future in which  $\AAA$ holds. 
The logic has also connectives that talk about \emph{space}, that is,
the shape of the edge-labelled trees that describe process
distributions.
 the formula $\amb n \AAA$ is satisfied by  
ambients  named $n$ 
 whose content  satisfies $\AAA$ (read on
trees: $\amb n \AAA$ is satisfied by the trees whose root has just a single edge
$n$ leading to a subtree that satisfies $\AAA$); 
 the formula $\AAA_1 | \AAA_2$ is satisfied by the processes 
that can be decomposed  into  parallel components 
$P_1$ and $P_2$ where each $P_i$ satisfies   $\AAA_i$ (read on
trees: $\AAA_1 | \AAA_2$ is satisfied by the trees that are the
juxtaposition of  
two trees that respectively satisfy the formulas $\AAA_1$ and $\AAA_2$);
the formula $\zero$ is satisfied by  the terminated   process $\nil$ (on
trees: $\zero $ is satisfied by the  tree consisting of  just the root node).

AL is quite different from standard modal logics.  First, the latter logics
do not talk about space. Secondly, they have more precise temporal
connectives.  The only temporal connective of AL talks about the
many-step evolution of a system on its own. In standard modal logics,
by contrast, the temporal connectives also talk about the potential
interactions between a process and its environment. For instance, in
the Hennessy-Milner logic~\cite{HeMi85}, the temporal modality
$\mess{\mu}.\AAA$ is satisfied by the processes that can perform the
action $\mu $ and become a process that satisfies $\AAA$.  The action
$\mu $ can be a reduction, but also an input or an output.


\bigskip

In this paper  we study the equivalence between
MA processes induced by the logic, written $\eqL$: we write $P\eqL Q$
if $P$ and $Q$ satisfy exactly the same formulas.
Our main goal is to understand how much the logic discriminates
between processes, i.e., to study the separating power of \eqL. 
We show that \eqL{} is a rather fine-grained
relation. 
Related to the problem of the equivalence 
 induced by the logic are  issues of decidability, that we also
 investigate. 

\medskip

%

The central technical device we rely on to analyse \eqL{} is
a characterisation  as a form of bisimilarity, that we call 
\emph{intensional bisimilarity} and  write \intbis. The bisimulation
game defining \intbis{} takes into account the interaction
possibilities of agents, and also includes clauses to observe the
spatial structure of processes, corresponding to the logical
connectives of emptyness, spatial conjunction, and ambient.
Intensional bisimilarity is to AL what standard bisimilarity is to
Hennessy-Milner logic. In particular, \intbis{} can be used to assess
separability and expressiveness properties of the modal logic it
captures. 
For instance,
the definition of \intbis{}
reveals that, 
in some cases, logical
observations are unable to distinguish between an agent
entering an ambient, and the same agent going in and out of this
ambient before finally entering it. We call this phenomenon
\emph{stuttering}.  Stuttering can be seen as the spatial counterpart
of the following `eta law' for the asynchronous
$\pi$-calculus~\cite{SW01}:
$$
a(x).\big(\,\overline{a}\langle x\rangle\,|\,a(x).P\,\big)
~=~ a(x).P
$$
\noindent (a similar equality also holds for communication in
MA).  Indeed, stuttering disappears when the asynchronous movements are
replaced by synchronous ones, as is the case, e.g., in the model of
Safe Ambients~\cite{LeSaTOPLAS}.

\medskip


Something worth stressing is  that our characterisation results
are established  on the full, public, MA calculus in which, as mentioned
earlier, terms need not be image-finite, and with respect to a
finitary  logic.   We are not aware of other results of this kind:
characterisation results for a bisimilarity with respect to a modal
logic in the literature (precisely, the completeness part of the characterisations)  rely either 
on an  image-finiteness
hypothesis for the terms of the language, or  on the presence of some
infinitary constructs (such as infinitary conjunctions) in the
syntax of the logic.
%
 Technically, the proof of our result is based on the definition of
some complex modal formulas.  To make it easier to understand our approach, we
 first present the main structure of the proof in a subcalculus without
 infinite behaviours; 
we then move to the full public 
MA calculus to show how replication is handled.
%
%
Our proof exploits two main technical notions. The first idea is to
introduce an induction principle on processes, that allows us to
provide an inductive characterisation of \intbis. We then introduce
modal formulas whose role is, intuitively, to establish that only finitely
many terms have to be taken into consideration when exploring the
outcomes of a given process.
%

\medskip



Exploiting \intbis, we relate logical equivalence with two important
equivalences for processes.  The first equivalence is the standard
extensional equivalence, namely barbed congruence ($\approx$).  Here
the main result is that logical equivalence is strictly finer. As
counterexamples to the inclusion $\wbc \subseteq \eqL $, we have found
three axiom schemata.  We do not know whether they are complete, that
is, if they exactly describe the difference between the two relations
on MA.

We then compare logical equivalence with a second relation, 
namely structural
congruence ($\equiv$),
 an
intensional and very discriminating equivalence.  
We establish an axiomatisation of logical
equivalence on a rather broad class of processes, called \MAIFsyn{}
(defined in~\ref{subsec:extint}). The definition of \MAIFsyn{} relies
on an image-finiteness constraint that is lighter than the usual
notion of image-finiteness in process calculi, 
because only certain subterms of processes
are required to give rise to finitely many reducts.  This subcalculus
is shown to be Turing complete in Section~\ref{s:undecidability}.
We are not aware of other axiomatisations of semantic equivalences
(defined by operational, denotational, logical, or other means) in
higher-order process calculi. Our result says that on \MAIFsyn, \eqL{}
almost exactly coincides with structural congruence, the only
difference being an `eta law' for communication of the form 
mentioned above. 
This axiomatisation does not hold in the  full MA, for instance 
 because of
the phenomenon of stuttering.

Communication in MA is asynchronous, in the sense that outputs have no
continuation. We show in~\ref{s:synchronous} that if asynchronous
communication is dropped in favour of synchronous communication, then
logical equivalence exactly coincides with structural congruence on
the synchronous version of \MAIFsyn. 

The comparisons reveal the intensional flavour of AL.  Although the
logic has operators for looking into the parallel structure of
processes, the intensionality of the logic was far from immediate,
essentially for two reasons. The first reason is that not all
syntactical constructions of MA are reflected in the logic, which
entirely lacks operators for capabilities, communications, and
replication. The second reason is that we adopt a weak interpretation
for reductions (i.e., we abstract from actions internal to the
processes); this makes it possible to handle infinite processes, but
at the same time entails a loss of precision when describing
properties of processes.  In such a setting it is therefore surprising
that $\eqL$ is actually so close to $ \equiv$, also because $\equiv $
is a very strong relation -- a few axioms are the only difference with
syntactic identity.



\medskip


Being very close to a syntactical description of processes,
the relation of structural congruence is decidable. 
As a consequence, in the subcalculus of MA where  we show that 
\eqL{} coincides with $\equiv$, we can also derive  
decidability for  \eqL{}.
However, the frontier with undecidability for \eqL{} is very subtle: we establish
undecidability of \eqL{} in the full calculus by encoding the halting
problem of a Turing machine. This boils down in our setting to specifing
Turing machines in Mobile Ambients and building a scenario where the
halting of a machine corresponds to the existence of \emph{reduction
  loops}, i.e., of processes $P$, $Q$ such that $P$ reduces to $Q$ and
$Q$ reduces to $P$.
This encoding is a challenging `programming
task', since the process must return to its initial state modulo \eqL;
this  is a demanding condition, since, as mentioned above, \eqL\ is a
rather strong relation.  For instance, one has to be
very precise in garbage collecting dead code during the execution of
the Turing machine.

\bigskip




\noindent\textbf{Other related work}
Although not directly related from a technical point of view, a work
worth mentioning is~\cite{dam:lics88}. In that work, models of
(enrichments of) relevant and linear logic are defined using Milner's
SCCS. In particular, the interpretation of implication is reminscent
of the definition of satisfaction for the guarantee operator ($\rtr$)
in AL. Dam however explicitely renounces giving sense to formulas that
talk about the structure of processes, as is the case in the Ambient
Logic.

As stated before, intensional bisimilarity is to AL what bisimilarity
is to Hennessy-Milner logic. Approximants of intensional bisimilarity,
that will be needed in our proofs of completeness, may also be
expressed in terms of Ehrenfeucht-Fra{\"\i}ss{\'e} games for spatial
logics, as shown in \cite{ghelli:fsttcs2004}. These equivalences are
standard devices to establish expressiveness results.  For instance,
they have been exploited to obtain adjunct elimination properties of
spatial logics in~\cite{caires::lozes::elim,lozes:tcs2005}.


This work is a revised and extended version of parts
of~\cite{Sangiorgi::ExtInt::01} and~\cite{Sedal}, precisely, those
parts  that deal with issues related to separability of AL.
A companion paper~\cite{Part1} studies
expressiveness issues.
By the time the writing of the present paper was completed, a few
papers have appeared that make use of results or methods presented
here. 
%
These are works that study the intensionality of spatial logics or
decidability properties. 
Works related to 
the intensionality of spatial logics 
include~\cite{CalcagnoCardelliGordon::validity}
where  the spatial logic is static,
and~\cite{caires::lozes::elim,caires:fossacs2004},
where the logic  is applied
to reason on calculi that feature a simpler notion of space, with a
strong interpretation of the temporal
modality.
A spatial logic for the $\pi$-calculus satisfying the property that
logical equivalence coincides with behavioural equivalence has been
studied in~\cite{hirschkoff::extensionalspatial04}. This logic is
defined by removing modal  operators like $\zero$ or spatial conjunction, and
keeping only `contextual' operators (guarantee and revelation
adjunct).
A similar result, but for a logic that includes spatial conjunction
and $\zero$, has been established for a process calculus encompassing
a form of distribution in~\cite{caires:vieira:express06}.
Works related to 
the decidability properties of Mobile Ambients
include~\cite{busi:zavattaro:tcs:2004,maffeis:phillips:tcs05}, that 
address questions of termination,
and~\cite{boneva:talbot:tcs:2005,busi:zavattaro:esop05}, that consider
reachability in \emph{syntactic} subcalculi of MA (in the sense that
these subcalculi are obtained by eliminating some syntactical
constructs). 
%
%
It can be noted that our analysis of decidability (in
Section~\ref{s:undecidability}) allows us to deduce a property in
terms of reachability: as discussed above, we establish that one
cannot detect the presence of \emph{reduction loops} (i.e., the
existence of processes $P$ and $Q$ that reduce to eachother). This in
particular entails undecidability of reachability.




\bigskip

\paragraph{Structure of the paper.}
We define the Mobile Ambients calculus and the Ambient Logic in
Section~\ref{sec:back}. Section~\ref{s:char} is devoted to the study
of intensional bisimilarity, \intbis. We show that \intbis{} is
included in logical equivalence, \eqL. Completeness, i.e., the
reverse inclusion, is first proved only for finite MA processes.  For
this, we need a certain number of expressiveness results about AL
from~\cite{Part1}, which are collected in~\ref{subsec:expressiveness}.
%
The completeness proof for the whole calculus is presented in
Section~\ref{s:inductive}, which completes our study of \intbis{} by
finally estabilishing that \intbis{} and \eqL{} coincide. 
%
%
The inductive characterisation of \intbis{} is given
in~\ref{subsec:bigind}, and the logical characterisation of the
outcomes of a process in~\ref{subsec:local:charac}. 
%
%
We compare \eqL{} with barbed congruence and structural congruence in
Section~\ref{s:axiom}. The subcalculus \MAIFsyn, on which we establish
an axiomatisation of \eqL, is also introduced here.
Subsection~\ref{s:synchronous} explains how our results are modified
when moving to synchronous Ambients.
We present our encoding of Turing machines into \MAIFsyn{} in
Section~\ref{s:undecidability}, and give concluding remarks in
Section~\ref{s:concl}.


\section{Background}\label{sec:back}
This section collects the necessary background for this paper. It
includes the Mobile Ambients calculus~\cite{CaGo98} syntax and
semantics, and the Ambient
Logic~\cite{Cardelli::Gordon::AnytimeAnywhere}.

\subsection{Syntax of Mobile Ambients}

We recall here the syntax of 
Mobile Ambients (MA)
(we sometimes also call
this calculus the Ambient calculus). In the calculus we study, only
names, not capabilities, can be communicated; this allows us to work
in an untyped calculus. 

The calculus is asynchronous; a
synchronous extension will be considered in Section~\ref{s:axiom}.
As in~\cite{Cardelli::Gordon::AnytimeAnywhere,CardelliSemiStructuredData,CardelliGhelli::ETAPS01},
the calculus has no restriction operator for creating new names.  

Table~\ref{ta:syn} shows the syntax. Letters $n,m,h $ range over
names, $x,y,z$ over variables; $\eta$ ranges over names and variables.
Both the set of names and the set of variables are infinite. The
expressions $\inamb \eta$, $\outamb \eta$, and $\openamb \eta$ are the
\emph{capabilities}.  Messages and abstractions are the
\emph{input/output} (I/O) primitives. A \emph{guard} is either an
abstraction or a capability. A process $P$ is \emph{single} iff there exists $P'$ such that either
  $P\,\equiv\,\capa.P'$ for some \capa{} or $P\,\equiv\,\amb{n}{P'}$
  for some $n$).

%
Abstraction is a binding construct, giving rise to the set of free
variables of a process $P$, written \fv{P}. We ignore syntactic
differences due to alpha conversion. We write 
\fn{P}  for the set of
(free) names of process $P$. A \emph{closed} process has no free
variable.  Unless explicitely stated, we  use $P,Q,\dots$ to range
over \emph{closed} processes in our definitions and results.
Substitutions, ranged over with $\sigma$, are partial functions from
variables to names. Given $\sigma$, we write $P\sigma$ to denote the
result of the application of $\sigma$ to $P$.  Given two processes $P$
and $Q$, we say that $\sigma$ is a closing substitution for $P$ and
$Q$ (in short, a closing substitution) if $P\sigma$ and $Q\sigma$ are
closed processes. We also introduce another notation: $P \sub n x$
stands for the capture avoiding substitution
 of variable $x$ with name
$n$ in $P$, and $P \sub n m$ stands for the process obtained by
replacing name $m$ with name $n$ in $P$.
Given $n$ processes $P_1,\dots,P_n$, we  sometimes write
$\Pi_{1\leq i\leq n}P_i$ for the parallel composition
$P_1\,|\ldots|\,P_n$. 

\emph{Process contexts} (simply called contexts) are processes
containing an occurrence of a special process, called the hole. We
 use $\C$ to range over process contexts, and $\C\fillhole{P}$ stands for
the process obtained by replacing the hole in $\C$ with $P$. 
Given two processes $P$ and $Q$, a \emph{closing context for $P$ and $Q$} (in
short, a closing context) is a context \C{} such that $\C\fillhole{P}$
and $\C\fillhole{Q}$ are closed processes.

\begin{center}
\begin{tabular}{cc}
$
\begin{array}[t]{lrll}
\multicolumn{3}{l}{
  h,k, \ldots n,m}
& {\mbox{{\it Names}}} 
\\[\tkp]
\multicolumn{3}{l}{
  x,y, \ldots}
& {\mbox{{\it Variables}}} 
\\[\tkp]
  \multicolumn{3}{l}{
   \eta}
& {\mbox{{\it Names $\cup$ Variables}}} 
\\[\tkp]
\\[\tkp]
 &&& {\mbox{\it Capabilities}}
\\[\tkp] 
\capa & ::=
   &  \ina \eta
   & \mboxB{enter}
\\[\tkp]
 & \midd
   &  \out  \eta
   & \mboxB{exit}
\\[\tkp]
 & \midd
   &  \open  \eta
   & \mboxB{open}
\end{array}
$ \qquad \qquad
\qquad& 
$\begin{array}[t]{lrll}
 &&& {\mbox{\it Processes}}
\\[\tkp] 
P,Q,R & ::=  
   & \nil
   &
   \mboxB{nil}
\\[\tkp]
 & \midd
   & P |  Q
   & \mboxB{parallel}
\\[\tkp]
 & \midd
   & !P
   & \mboxB{replication}

\\[\tkp]
 & \midd
   & \capa . P
   & \mboxB{prefixing}
\\[\tkp]
 & \midd
   & \amb{ \eta}{P}
   & \mboxB{ambient}
\\[\tkp]
 & \midd
   & \msg{ \eta}
   & \mboxB{message}
\\[\tkp]
 & \midd
   & \abs{x} P
   & \mboxB{abstraction}
\end{array} 
$
\end{tabular}
\end{center} 

\label{ta:syn}

Processes with the same internal structure are identified. This is
expressed by means of the \emph{structural congruence relation},
$\equiv$, the smallest congruence such that the following laws hold:

\begin{mathpar}

  P | \nil ~\equiv~ P
  \and
  P | Q ~\equiv~ Q | P
  \and
  P | (Q|R)~ \equiv~ (P | Q) | R
  \\
    !P ~ \equiv ~ !P | P
    \and
    !\nil ~ \equiv ~ \nil
    \and
    !(P|Q) ~ \equiv ~ !P | !Q
    \and
    !!P ~ \equiv ~ !P
\end{mathpar}
%




As a consequence of the results presented in~\cite{Dal01}, which works
with a richer calculus than the one we study, we have:
\begin{thm}\label{thm:equivdeci}
\label{t:dz}
$\equiv$ is decidable.
\end{thm} 
%


\begin{defi}[Finite process]\label{def:finitesingle}\mbox{}
 A process $P$ is \emph{finite} iff there exists a process $P'$
    with no occurrence of the replication operator such that
    $P\,\equiv\, P'$.
\end{defi}

\subsection{Operational Semantics}

The semantics of the calculus is given by a reduction relation
$\longrightarrow$. We shall sometimes use the phrase
`$\tau$-transitions' to refer to $\longrightarrow$ transitions. The
corresponding rules are given in Table~\ref{t:red_rules}. The
reflexive and transitive closure of $\longrightarrow $ is written
$\Longrightarrow $.

\begin{table}
\label{t:red_rules}
\begin{mathpar}
   \shortaxiomC{ \open n . P |
\amb n{ Q} \longrightarrow 
P |  Q
}
{Red-Open} 
\and
  \shortaxiomC
{ \amb n {\ina m . P_1 | P_2} | 
\amb m{ Q} \longrightarrow 
\amb m{ \amb n { P_1 | P_2}| Q
}}{Red-In}
\and
 \shortaxiomC
{
\amb m{ \amb n {\out m . P_1 | P_2}| Q}
\longrightarrow 
 \amb n {P_1 | P_2} | 
\amb m{ Q} 
}
{Red-Out}  
\and
 \shortaxiomC
{ \msg \eta |
\abs x  P \longrightarrow 
P \sub \eta x
} 
{Red-Com}
\and
\shortinfruleN
{
 P  \longrightarrow  P'
}
{
P |Q  \longrightarrow  P' |Q
}
{Red-Par}
\and
\shortinfruleN
{
 P   \longrightarrow  P'
}
{
\amb n P   \longrightarrow  \amb n {P'}
}
{Red-Amb}
\and
\shortinfruleN
{
 P \equiv P' \andalso P'  \longrightarrow  P'' \andalso P'' \equiv P'''
}
{
P   \longrightarrow  P'''
}
{Red-Str}
\end{mathpar}
\caption{The rules for reduction}
\Mybar  
\end{table}

\emph{Behavioural equivalence} is defined using reduction and
observability predicates $\Dwa_n$ that indicate whether a process can
liberate an ambient named $n$: formally, $P \Dwa_n$ holds if there are
$P',P''$ such that $P\Longrightarrow \amb n {P'}| P''$.

\renewcommand{\RR}{\ensuremath{\mathcal{R}}}

\begin{defi}[barbed  congruence, \cite{MiSa92,LeSa00full}]
\label{d:bb}
A symmetric relation \RR{} between processes is a \emph{barbed
  bisimulation} if $P \RR{} Q$ implies:
\begin{enumerate}[(1)]
\item   whenever $P \Longrightarrow  P'$,
  there exists $Q'$ such that
 $Q {\Longrightarrow}Q'$ and $P' \RR  Q'$;
 \item   for each name  $n$,    
 $P \Dwa_n$ iff     $Q \Dwa_n$.
\end{enumerate}
{\em Barbed  bisimilarity}, written $\wbb$,
is the largest barbed bisimulation. 
Two processes $P$ and $Q$ are {\em barbed
congruent},
written
$P \wbc  Q$,
 if 
$\C\fillhole{P} \wbb   \C\fillhole{Q}$ for all closing contexts $\C$.
\end{defi}

\subsection{Ambient Logic}

The Ambient Logic  (AL),  is presented in 
Table~\ref{ta:for}).
We use 
 an infinite set of \emph{logical
  variables}, ranged over with $x,y,z$; $\eta$ ranges over names and
variables.  
(We can use the same syntax as for variables and names of the Ambient
calculus, since formula and process terms are separate.)
We use $\A,\B,\dots,\F,\F',\dots$ to range over
formulas.

\begin{table*}[ht]
\begin{center}
$
\begin{array}[t]{lrlll}
\AAA & ::=
   &  \ltrue 
   & \mboxB{true}
   & \mbox{classical logic}
\\[\tkp]
&\midd&  \neg\AAA
   & \mboxB{negation}
\\[\tkp]
&\midd&  \AAA \lor \BBB
   & \mboxB{disjunction}
\\[\tkp]
&\midd& \all x \AAA
   & \multicolumn{2}{l}{\mboxB{universal quantification over names}}
   \\[\tkp]
   & \midd &
   \diamond\AAA
&\mboxB{sometime}
&\mbox{temporal and spatial connectives}
\\[\tkp]
&\midd& \zero
&\mboxB{void}
\\[\tkp]
&\midd& \amb\eta\AAA
&\mboxB{edge}
\\[\tkp]
&\midd& \AAA | \BBB
&\mboxB{composition}
\\[\tkp]
& \midd &
\at{\AAA}{\eta}
&\mboxB{localisation}
&\mbox{logical adjuncts}
\\[\tkp]
&\midd &\AAA~\rtr~\BBB
   & \mboxB{linear implication}\qquad
\end{array}
$
\end{center} 
\caption{The syntax of logical formulas}
\Mybar  
\label{ta:for}
\end{table*} 

 The logic has the propositional connectives, $\ltrue, \non \AAA, \AAA
 \lor \BBB$, and universal quantification on names, $\forall x.\,
 \AAA$, with the standard logical interpretation.  The temporal
 connective, $\sometime \AAA$ is considered with a weak
 interpretation. The spatial connectives, $\zero$, $\AAA | \BBB$, and
 $\amb \eta \AAA$, are the logical counterpart of the corresponding
 constructions on processes. $\AAA \rtr \BBB $ and $\at{\AAA}{\eta}$
 are the adjuncts of $\AAA | \BBB$ and $ \amb \eta \AAA$, in the sense
 of being, roughly, their inverse (see below).  $\AAA\{n/x\}$ is the
 formula obtained from $\AAA$ by substituting variable $x$ by name
 $n$.  A formula without free variables is \emph{closed}. Along the
 lines of the definition of process contexts, we define formula
 contexts as formulas containing an occurrence of a special \emph{hole
 formula}.

We use $\A\fillhole{\cdot}$ to range over formula contexts; then
$\A\fillhole{\B}$ stands for the formula obtained by replacing the
hole in $\A\fillhole{\cdot}$ with $\B$.

\begin{defi}[Satisfaction]
\label{d:satisfaction}
The satisfaction relation is defined between closed processes and
closed formulas as follows:
$$
\begin{array}{lcl}
\satDS P \true &\defiDS & \mbox{always true} 
\\
        \satDS P{\all x \AAA } & \defiDS & \mbox{for any $n$, } \satDS
        P \AAA \{ n /x\} 
        \\
\satDS P{ \myneg \AAA } & \defiDS & \mbox{not }  \satDS P{  \AAA }
\\
        \satDS P{ \AAA_1 | \AAA_2 } & \defiDS & 
                \exists P_1,P_2 \mbox{ s.t. } 
                P \equiv {P_1|P_2}
                \mbox{ and } \satDS{P_i}{\AAA_i},~i=1,2
        \\
        \satDS P{\AAA \orr \BBB } & \defiDS & \satDS P{\AAA } \mbox{ or }
        \satDS P{\BBB } 
\\
        \satDS P{\amb n \AAA } & \defiDS & 
                \exists P' \mbox{ s.t. }P \equiv \amb n {P'}
                \mbox{ and } \satDS{P'}\AAA
        \\
\satDS P{\zero } & \defiDS & P\equiv\nil
\\

        \satDS P{\sometime \AAA  } & \defiDS & 
                \exists P' \mbox{ s.t. }
                 P \Rar  P' \mbox{ and } \satDS{P'}\AAA 
        \\
        \satDS P{\at{\AAA}{n}   } & \defiDS & 
                \satDS{\amb n P } \AAA
\\
                \satDS P{\AAA \limp \BBB    } & \defiDS & 
                \forall R,~ R\sat\AAA \mbox{ implies } \satDS{P|R} \BBB
\end{array}
$$
\end{defi}

The logic in~\cite{Cardelli::Gordon::AnytimeAnywhere} has also a
\emph{somewhere} connective, that holds of a process containing, at
some arbitrary level of nesting of ambients, an ambient whose content
satisfies $\AAA$.  
For the sake of simplicity, we omit this connective, but we believe
that the addition of this connective would not change the results in
the paper (in particular Theorem~\ref{t:soundnessLogBis} can be
adapted easily).

\begin{lem}[\cite{Cardelli::Gordon::AnytimeAnywhere}]
  If $P \equiv Q $ and $\satDS P \AAA$, then also $\satDS Q \AAA$.
\end{lem}

We give $\lor $ the least syntactic precedence, thus
$\AAA_1 \rtr \AAA_2 \lor \AAA_3$ reads $(\AAA_1 \rtr \AAA_2) \lor
\AAA_3$, and $\AAA_1 \rtr (\sometime \AAA_2 \lor \sometime \AAA_3)$
reads $\AAA_1 \rtr ((\sometime \AAA_2) \lor (\sometime \AAA_3))$.
We shall use the following standard  duals of disjunction and
universal quantification:
\begin{mathpar}
  \AAA \land B ~\defiDS~  \non (\non \AAA \orr \non \BBB)
  \and
  \thereis x \AAA ~\defiDS~ \non \all x \non \AAA   
\end{mathpar}




\begin{defi}[Logical equivalence]
\label{d:eqLOGIC}
For processes $P$ and $Q$, we say that $P$ and $Q$ are \emph{logically
  equivalent,} written $P\, \eqL\, Q$, if for any closed formula $\AAA$
it holds that $\satDS P \AAA $ iff $\satDS Q \AAA$.
\end{defi} 

The remainder of this paper is devoted to the study of $\eqL$ on MA
and on some subcalculi of MA.


\section{Intensional bisimilarity}\label{s:char}

In order to be able to carry out our programme for $\eqL$, as
discussed in the introduction, we 
look for a co-inductive characterisation of this relation, as a form
of labelled bisimilarity.
%
Before introducing the bisimilarity relation, we need to define
labelled transitions on \MA, and a few derived relations such as
the
\emph{stuttering} relation. 

\subsection{Definitions}

\subsubsection{Labelled transitions and stuttering}

\begin{defi}
\label{d:statt}
Let $P$ be a closed process. We write: 
\begin{enumerate}[$\bullet$]
\item
$P \arr{\capa} P'$, where $\capa$ is a capability, if 
$P \equiv  \capa.P_1|P_2$ and $P' = P_1|P_2$.

\item $P \arr{\msg n} P'$ if 
$P \equiv  \msg n|P'$.

\item $P \arr{?n} P'$ if 
$P \equiv  \abs x {P_1}|P_2$ and $P' \equiv  P_1\sub n x|P_2$. 

\item $P \Ar{\mu}P'$, where $\mu$ is one of the above labels, if $P
  \Longrightarrow \arr\mu \Longrightarrow P'$ 
(where $
  \Longrightarrow \arr\mu \Longrightarrow$ is relation composition). 

\item {\bf (stuttering)}  $P  \Stat{M_1}{M_2} P' $ if 
there is $i\geq 1$ and processes $P_1 , \ldots, P_i$ with $P = P_1 $
and $P' = P_i$  such that 
$P_r \Ar {M_1}\Ar {M_2}P_{r+1}$ for all $1 \leq r < i$. 

\item  Finally, $\Rcap$ is a convenient notation for compacting statements
  involving capability transitions. $\Rcapar{\inamb{n}}$ is
  $\Stat{\outamb{n}}{\inamb{n}}$; similarly $\Rcapar{\outamb{n}}$ is
  $\Stat{\inamb{n}}{\outamb{n}}$; and $\Rcapar{\openamb{n}}$ is
  $\Rar$.
\end{enumerate}
\end{defi}

We discuss in Example~\ref{e:stuttloop} below why
stuttering is needed to capture logical equivalence in \MA.

\subsubsection{Intensional bisimilarity, $\bisMOD$}

We present here our main labelled bisimilarity, 
   \emph{intensional bisimilarity}, 
written $\bisMOD$. This relation  will be used to
capture the separating power of $\eqL$.


Intuitively, the definition of $\bisMOD$ is based on the observations
made available by the logic either using built-in operators or through
derived formulas for capabilities (see below).



\begin{defi}
\label{d:bisMOD}
A symmetric relation \RR{} on closed processes is an \emph{intensional
  bisimulation} if $P\RR Q$ implies:
\begin{enumerate}[(1)]
\item
\label{cMOD:par}
If $P \equiv P_1| P_2$ 
then there are $Q_1,Q_2$ such that 
$Q \equiv  Q_1|Q_2 $ and
$P_i \RR Q_i$, for $i=1,2$.

\item
\label{c:nil}
 If $P \equiv \nil $ then  $Q \equiv \nil$.

\item
\label{c:tau}
 If $P \longrightarrow P'$ then there is $Q'$ such that 
$Q\Longrightarrow Q'$ and $P' \RR Q'$.
\item
\label{cMOD:in}
 If 
$P \arr{\ina n} P'$
 then 
there is $Q'$ such that
 $Q \Ar{\ina n } \StatOI n Q' $
and $P' \RR Q'$.

\item
\label{cMOD:out}
 If $P \arr{\out n } P'$ then 
there is $Q'$ such that 
$Q \Ar{\out n } \StatIO n Q' $
and $P' \RR Q'$.

 \item
\label{cMOD:open}
 If $P\arr{ \open n } P'$ then 
there is $Q'$ such that $Q \Ar{\open n } Q' $
and $P' \RR Q'$.

\item
\label{cMOD:msg}
 If $P \arr{\msg n} P'$ then 
there is $Q'$ such that 
$Q \Ar{\msg n}  Q' $ and $P' \RR Q'$.

\item
\label{c:absMOD} If $P   \arr{? n}  {P'}$
then
there is $Q'$ such that $Q | \msg n \Longrightarrow Q'
$ and  $P' \RR Q'$.

\item 
\label{cMOD:amb}
If $P \equiv \amb n {P'}$ then there is $Q'$ such that $Q 
\equiv \amb n {Q'}$ and $P' \RR Q'$.
\end{enumerate} 
\emph{Intensional bisimilarity}, written $\bisMOD$, is the largest
intensional bisimulation.
The definition of $\bisMOD$ induces a relation $\bisMODo$, defined on
open terms by saying that $P\bisMODo Q$ iff for any closing
substitution $\sigma$, $P\sigma\bisMOD Q\sigma$.
\end{defi} 

The definition of $\bisMOD$ has (at least) three intensional clauses,
namely \reff{cMOD:par}, \reff{c:nil} and \reff{cMOD:amb}, which allow
us to observe parallel compositions, the terminated process, and
ambients. These clauses correspond to the intensional connectives
`$|$', `$\zero$' and `$\amb n\cdot$' of the logic.  The clause
\reff{c:absMOD} for abstraction is similar to the input clause of
bisimilarity in asynchronous message-passing calculi \cite{AmCaSa98}.
This is so because communication in MA is asynchronous (see also
Subsection~\ref{s:synchronous} below).
Note that, using notation $\Rcap$ introduced above, items 4, 5, and 6 can be
replaced by the following one: 
\begin{enumerate}[$\bullet$]
\item
if $P\rcap P'$, then there is $Q'$ such that
$Q\Ar{\capa}\Rcap Q'$ and $P'\RR Q'$.
\end{enumerate}
 As we have pointed out above, stuttering is used to capture some
transitions of processes that the logic cannot detect. It gives rise
to particular kinds of loops, that we illustrate in the following
example.

\begin{exa}[Stuttering Loop]\label{e:stuttloop}
Consider the processes
\[
\begin{array}{rcl}
P&\defiDS& !\openamb{n}.\inamb{n}.\outamb{n}.\inamb{n}.\outamb{n}.\amb{n}{\nil}
 | \amb{n}{\nil} 
\\
Q&\defiDS& 
!\openamb{n}.\inamb{n}.\outamb{n}.\inamb{n}.\outamb{n}.\amb{n}{\nil}
| \inamb{n}.\outamb{n}.\amb{n}{\nil}\,.  
\end{array}
\]

We have the following loop, modulo stuttering:
\begin{mathpar}
  { P~\Stat{\inamb n}{\outamb n}~Q~\Stat{\inamb n}{\outamb n}~P
    \enspace. }
\end{mathpar}

The existence of such pairs of processes that reduce one to each other
modulo stuttering will play an important role in the axiomatization of
$\eqL$. We call such a situation a loop.

%
It holds that   $P\not \bisMOD Q$; however, since 
$P~\Stat{\inamb n}{\outamb n}~Q~\Stat{\inamb n}{\outamb n}~P$,
we have 
$$
\outamb{n}.P~\bisMOD~\outamb{n}.Q\,.
$$
Actually, $\outamb{n}.P~\wbc~\outamb{n}.Q.$, that is, these two
processes are extensionally equivalent, and they are also equated by
the logic (i.e., $\outamb{n}.P\,\eqL\, \outamb{n}.Q$).  But they would
not be intensionally bisimilar without the stuttering relations.

The reason for this peculiarity is that, intuitively, these processes
have the same behaviour in any testing context.  To see why the extra
capabilities of $Q$ do not affect its behaviour, consider a reduction
involving $\outamb n.P$, of the following shape:
\[ \amb n {\amb m {\outamb n.P|R}}  \longrightarrow \amb n \nil | \amb
m {P|R}\,. \] 
Process $\outamb n.Q$ can match this transition using three reductions:
\[
\begin{array}{rcl}
  \amb n {\amb m {\outamb n.Q|R}} & \longrightarrow& \amb n \nil | 
  \amb m{\inamb n. \outamb n .n[\nil]|Q'|R}\\[\tkpPOPL]
& \longrightarrow&  
\amb n {\amb m {\outamb n.n[\nil]|Q'|R}}
\\[\tkpPOPL]
&  \longrightarrow&  \amb n \nil | \amb m {P|R}\,,
\end{array}
\]
where $Q'$ is
$!\openamb{n}.\inamb{n}.\outamb{n}.\inamb{n}.\outamb{n}.\amb{n}{\nil}$.
Conversely, the process $\outamb n.Q$ may be involved in the following
scenario:
\[ \amb n {\amb m {\outamb n.Q|R}}  \longrightarrow \amb n \nil | \amb
m {Q|R}\,,
\]
and the process $\outamb n.P$ can mimic this reduction.

If we set
$Q'=!\openamb n.\inamb n.\outamb n.\inamb n.\outamb n.n[\nil]$, we
have
\[
\begin{array}{rcl}
  \amb n {\amb m {\outamb n.P|R}} & \longrightarrow& \amb n \nil | \amb m{\amb n \nil
    |Q'|R}
\\[\tkpPOPL] & \longrightarrow&  
\amb n \nil | \amb m {Q'|\inamb n.\outamb n.\inamb n.\outamb
  n.n[\nil]|R}
\\[\tkpPOPL] & \longrightarrow&  
\amb n {\amb m {Q'|\outamb n.\inamb n.\outamb n.n[\nil]|R}}
\\[\tkpPOPL] & \longrightarrow&  
 \amb n \nil | \amb m {Q|R}\,.
\end{array}
\]
By~contrast,  stuttering does not show up  in   Safe
Ambients \cite{LeSa00full}, where  movements are achieved by means of 
synchronisations between a  capability and a \emph{co-capability}, and
alike models.
\end{exa}

The following result is an easy consequence of the definition of
$\bisMOD$:
\begin{lem}
\label{l:MODeq}
 $\bisMOD$ is an equivalence relation.
\end{lem} 

 \begin{proof}
The only point worth mentioning  is that, for transitivity,
to handle clause (8),
one first needs to prove 
that $\bisMOD$ is preserved by parallel compositions
 with messages (which is anyhow straightforward).
 \end{proof} 

However, it is not obvious that $\bisMOD$ is preserved by all
operators of the calculus, due to the fact that $\bisMOD$ is,
intrinsically, higher-order.  Formally, $\bisMOD$ is not higher-order,
in that the labels of actions do not contain terms.
Clause~\reff{c:tau} of Definition~\ref{d:bisMOD}, however, involves
some higher-order computation, for a reduction may involve movement of
terms (for instance, if the reduction uses rules \trans{Red-In} or
\trans{Red-Out}).  This, as usual in higher-order forms of
bisimilarity, complicates the proof that bisimilarity is preserved by
parallel composition.



\subsection{Congruence}

In this section, we establish congruence of intensional bisimilarity,
using an auxiliary relation.

\subsubsection{Syntactical relation, $\bisE$}

Our proof of congruence  makes use of a second bisimilarity,
$\bisE$, that,  by construction,  is preserved by all operators
of the calculus, and that is defined as follows:




\begin{defi}
\label{d:bisEwithEqual}
A symmetric relation on processes \RR{} is a \emph{syntax-based
  intensional bisimulation}  if $P\RR Q$ implies:
\begin{enumerate}[(1)]
\item
\label{cE:par}
If $P \equiv P_1 | P_2$
then there are $Q_s$ ($s=1,2$) such that 
$Q \equiv  Q_1 |Q_2 $ and for all $s$ 
$P_s \RR Q_s$.

\item 
\label{cE:in}
If $P\equiv \capa. P'$ then 
there are $Q',Q''$ such that
\begin{enumerate}
\item $Q \equiv \capa. Q'$, 
\item $Q'\Rcap Q''$, and
\item $P' \RR Q'$.
\end{enumerate}
 
\item 
\label{cE:msg}
If $P \equiv \msg n$ then 
$Q \equiv \msg n $.

\item
\label{c:abs} If $P \equiv \abs x {P'}$
then
there is $Q'$ such that
\begin{enumerate}
\item
 $Q \equiv \abs x {Q'}$ and 
\item for all $n$
there is $Q''$ such that
$ \msg n | Q \Longrightarrow Q'' $ and $P' \sub n x  \RR Q''$.
\end{enumerate}
 
\item If $P \equiv \amb n {P'}$ then there is $Q'$ such that $Q 
\equiv \amb n {Q'}$ and $P' \RR Q'$.
\end{enumerate} 
$\bisE$ is the largest syntax-based intensional bisimulation. Given two
open terms $P$ and $Q$, we say that $P\bisEo Q$ holds iff for any
closing substitution $\sigma$, $P\sigma\bisE Q\sigma$.
\end{defi}

Clause~(\ref{c:abs}) is typical of asynchronous calculi, as in
clause~(\ref{c:absMOD}) of Definition~\ref{d:bisMOD}.
The differences between the definitions of $\bisMOD $ and $\bisE $ are
the following. First, labelled transitions are replaced by structural
congruence in the hypothesis of the corresponding clause. Second,
clause \reff{c:tau} about reductions of related processes is removed.
%
Note that a  clause for the process \nil{} is not necessary (see
Lemma~\ref{l:nil-bis} below).

Transitivity of $\bisE$ is not obvious, because it is not immediate
that $\bisE$ is preserved under reductions (there is no clause for
matching $\tau$-transitions, and reductions (i.e., relation
$\Longrightarrow $) are used in a few places, such as the stuttering
relation in the clauses for movement.

We shall prove that $\bisMOD $ and $\bisE$ coincide
(Corollary~\ref{c:MOD-E} below). Thus, transitivity of $\bisE$ will
hold because of $\bisMOD$'s transitivity, and conversely, congruence
of $\bisE$ will ensure congruence of $\bisMOD$.
This proof method, which exploits an auxiliary relation that is
manifestly preserved by the operators of the calculus but that is not
manifestly preserved under reductions, brings to mind Howe's proof
technique for proving congruence of bisimilarity in higher-order
languages~\cite{How96}.  In our case, however, the problem is simpler
because of the intensional clauses~\reff{cMOD:par} and~\reff{c:nil} of
the bisimilarity and because MA is not a fully higher-order calculus:
terms may move during a computation, but they may not be copied as a
consequence of a movement.  We may say that MA is a \emph{linear}
higher-order calculus (indeed the congruence of $\bisMOD$ could also
be proved directly, with a little more work).

\medskip


In order to establish congruence of $\bisE$, we introduce an important
equality between processes, that plays a technical role here but will
also be used when characterising logical equivalence in
Section~\ref{s:axiom}.


\begin{defi}[Eta law, $\equiv_E$]
\label{d:etalaw}
The \emph{eta law} is given by the following equation:
\[
\abs x ( \abs x P| \msg x ) = \abs x P\,.
\]
We use the eta law to define the following three relations:
\begin{enumerate}[$\bullet$]
\item
 \etarew{} is the eta law oriented from left to right;
that is, $P  \etarew Q$ holds if $Q$ is obtained from $P$ by 
applying the  eta law once, from left to right, to one 
 of its subterms (modulo $\equiv$).

\item $\etarew^*$ stands for the reflexive, transitive closure of
  \etarew; 
\item $\equiv_E$ is the smallest congruence satisfying the laws
of $\equiv$ plus the eta law.
\end{enumerate} 
\end{defi}
%


In the lemma below, we write $P\etarewH P'$  if $P\etarew P'$ and
this represents  a top-level rewrite step, i.e.,
    we do not  rewrite under capabilities and input prefixes. 
Similarly, $\etarewH^*$ is the reflexive and transitive closure of  $\etarewH $.


\begin{lem}\label{l:eta_nf}
Let $\mathcal R$ stand for $\etarew$ or $\etarewH$. We say that
\begin{enumerate}[\em(1)]
  \item $\mathcal R$ is confluent up to $\equiv$, that is, for all
    $P,Q,R$ such that $P{\mathcal R^*} Q$ and $P{\mathcal R^*} R$, 
there is $Q',R'$ such that $Q{\mathcal R^*} Q'$, $R{\mathcal R^*} R'$ and $Q'\equiv R'$.
  \item $\mathcal R$ is terminating, that is $\mathcal R^*$ is a well-founded order.
  \end{enumerate}
\end{lem}  

We call the \emph{eta normal form} of $P$ (the \emph{head eta normal form} of $P$, respectively) the unique normal form,  up to $\equiv$, of $\etarew$ (of $\etarewH$, respectively).


 



\begin{rem}[Eta law and stuttering]
  The eta law 
  expresses a form of stuttering (in communication, as opposed to
  stuttering in movements -- see Definition~\ref{d:statt}).
  The logic being insensitive to both forms of stuttering, we have to
  reason modulo the eta law.
\end{rem}

\medskip

We now present some results that are needed to prove congruence of
$\bisE$.

\begin{lem}
\label{l:nil-bis}
If $\nil \bisE Q $ then $Q \equiv \nil$.
\end{lem} 
\begin{proof}
  Suppose $Q\equiv \nil$ does not hold. This means that there exists
  $Q', Q''$ s.t.\ $Q \equiv Q' | Q''$, with $Q'$ is of the form $\abs
  x R$, $\msg \MM$, $M . R$, or $\amb n R$. Then by applying the
  corresponding clause in the definition of $\bisE$, we deduce
  $Q\not\equiv\nil$, i.e., a contradiction.
\end{proof} 

\begin{lem}
\label{c:equiv-bis}
$\equiv_E\, \subseteq\, \bisE$ and $\equiv_E\bisE\equiv_E\,\, \subseteq\,\, 
\bisE$.
\end{lem}
 
\begin{proof}
Straightforward from the definition of $\bisE$.
\end{proof} 

If $\mathcal R$ is a binary relation on processes, we note $\mathcal
R\sub n m$ for the relation defined as
$\{(P\sub n m,Q\sub n m).\, (P,Q)\in \mathcal
R\}$. 

\begin{lem}\label{l:rename_bis}
 \label{l:subst_bis}
If $\mathcal{R}$ is a $\bisE$-bisimulation, then for any $n, m$,
$\mathcal{R}\sub n m$ is a $\bisE$-bisimulation.
\end{lem}

\begin{proof}
  Since $\tau$ transitions are not tested in $\bisE$, substitution is
  not mentioned in Def.~\ref{d:bisEwithEqual}. All clauses of the
  latter definition are obviously stable by substitution.
\end{proof}



\begin{lem}
\label{l:pres-Bis}
For any possibly open processes $P$ and $Q$, if
 $P \bisEo Q$ then $\C\fillhole{P} \bisEo
\C\fillhole{Q}$, for all contexts $\C$.
\end{lem} 

\begin{proof}
  By induction on $\C$, using the definition of $\bisE$.
\end{proof} 

\smallskip

To prove that $\bisMOD$ and $\bisE$ coincide, the main result we need
is that $\bisE$ is preserved under reductions:

\begin{lem}
\label{l:tau-action}
Suppose $P \bisE Q$ and $P \longrightarrow P' $. Then   there is $Q'$
such that $Q \Longrightarrow Q'$ and $P' \bisE Q'$.
\end{lem} 

\proof
By induction on the depth of the derivation proof of $P
\longrightarrow P' $. We proceed by case analysis on the last rule
used in the derivation. 

\begin{enumerate}[$\bullet$]
\item Rule \trans{Red-struct}:  
\[
\begin{array}{cr}
\Ainfrule
{
 P \equiv P_1 \andalso P_1  \longrightarrow  P_2 \andalso P_2 \equiv P_3
}
{
P   \longrightarrow  P_3
}
{}
\end{array}
 \]
By Lemma~\ref{c:equiv-bis}, $P_1 \bisE Q$; by induction $Q \Longrightarrow
Q' \bisE P_2$; again by Lemma~\ref{c:equiv-bis}, $Q' \bisE P_3$. 

\item  Rule \trans{Red-Par}:
 \[
\begin{array}{cr} 
\Ainfrule
{
 P_1  \longrightarrow  P_1'
}
{
P_1 |P_2  \longrightarrow  P_1' |P_2
}
{}
\end{array}
 \]
By  definition of $\bisE$ there are $Q_i$ such that 
 $Q \equiv Q_1 | Q_2 $  and $P_i \bisE
Q_i$. Then  we conclude, using induction and Lemma~\ref{l:pres-Bis}.

\item Rule \trans{Red-Amb}: Use induction and Lemma~\ref{l:pres-Bis}.

\item Rule \trans{Red-Com}:
Immediate by clauses \reff{cE:par},  \reff{cE:msg}, 
and \reff{c:abs} of Definition~\ref{d:bisEwithEqual}.

\item Rule \trans{Red-Open}:
 \[
\begin{array}{cr} 
\Aaxiom
{ \open n . P_1 |
\amb n{ P_2} \longrightarrow 
P_1 |  P_2
}
{}
\end{array}
 \]
By definition of $\bisE$, $Q \equiv \open n . Q_1 |
\amb n{ Q_2}$, and for some $Q_1'$ with $Q_1 \Longrightarrow
Q_1'$, we have:  $P_2 \bisE Q_2$, $P_1 \bisE Q_1'$.
We also have $Q \Longrightarrow Q_1' | Q_2$. Using
Lemma~\ref{l:pres-Bis},
 we derive $ P_1 |  P_2 \bisE Q_1' | Q_2$, which
concludes the case. 

\item Rule \trans{Red-In}:
 \[
\begin{array}{cr} 
\Aaxiom
{
 \amb n {\ina m . P_1 | P_2} | 
\amb m{ P_3} \longrightarrow 
\amb m{ \amb n { P_1 | P_2}| P_3
}}{}
\end{array}
 \]
By definition of $\bisE$,
$Q \equiv
 \amb n {\ina m . Q_1 | Q_2} | 
\amb m{ Q_3}$, and there exists $Q_1'$ such that $Q_1
 \StatOI m Q_1' $ and we have:  $P_2 \bisE Q_2$, 
$P_3 \bisE Q_3$, and 
$P_1 \bisE Q_1'$.

We also have $Q \Longrightarrow 
\amb m{ \amb n { Q'_1 | Q_2}| Q_3
}$. Using Lemma~\ref{l:pres-Bis}, we derive 
$$
\amb m{ \amb n { P_1 | P_2}| P_3
} \bisE 
\amb m{ \amb n { Q'_1 | Q_2}| Q_3
}, $$ which
concludes the case. 
\item Rule \trans{Red-Out}: similar to the previous case.\qed
\end{enumerate}\smallskip

\begin{cor}
\label{c:tau-action}
Suppose $P \bisE Q$ and $P \Longrightarrow P' $. Then there is $Q'$
such that $Q \Longrightarrow Q'$ and $P' \bisE Q'$.
\end{cor}
 
\begin{proof}
By induction on the number of transitions in  $ P \Longrightarrow P'$, 
using Lemma~\ref{l:tau-action} for the inductive case.  
\end{proof}


          
\begin{lem}
\label{l:unCOMP}~
\begin{enumerate}[$-$]
\item
$\capa . P \bisMOD Q$ implies $Q \equiv \capa. Q'$, for some $Q'$.

\item
$\msg n \bisMOD Q$ implies $Q \equiv \msg n$.

\item
$\abs x P \bisMOD Q$ implies $Q \equiv \abs x {Q'}$, for some $Q'$.
\end{enumerate}
 \end{lem} 

\begin{proof}
  In every case, we suppose by contadiction that $Q \equiv Q_1 | Q_2$
  where none of the $Q_i$s is structurally congruent to $\nil$. Then
  $P$ and $Q$ can be distinguished using the clauses of $\bisMOD$ for
  parallel composition and $\nil$, which means a contradiction.
  
  Therefore, $Q$ is single (it has only one component), and we can
  conclude using the appropriate clause of the definition of $
  \bisMOD$ in each case.
\end{proof}

\begin{lem}
\label{l:MOD-E} 
$\bisMOD\, \subseteq\, \bisE$.
\end{lem} 
\begin{proof}
By proving that $\bisMOD $ is a $\bisE$-bisimulation. 
The proof is easy, using Lemma~\ref{l:unCOMP}. 
\end{proof}

\begin{lem}
\label{l:E-MOD} 
$\bisE\, \subseteq\, \bisMOD$.
\end{lem} 
\begin{proof}
By proving that $\bisE $ is a $\bisMOD$-bisimulation. 
We need Lemma~\ref{l:pres-Bis} (precisely, the fact that $\bisE$ is
preserved by parallel composition), Lemma~\ref{c:equiv-bis}, 
 Corollary~\ref{c:tau-action}, and
Lemma~\ref{l:nil-bis}.   
\end{proof}

\begin{cor}
\label{c:MOD-E}
Relations $\bisMOD$ and $\bisE$ coincide.
\end{cor}

\begin{cor}
\label{c:bisEcong}
Relations $\bisMODo $ and $\bisEo $ are congruence relations.
\end{cor}

\begin{proof}
Follows from Corollary~\ref{c:MOD-E}, and Lemmas~\ref{l:MODeq} and \ref{l:pres-Bis}  
\end{proof}









\subsection{Expressiveness results}
\label{subsec:expressiveness}

In this subsection we recall some expressiveness results for AL. These
results state the existence of formulas capturing some nontrivial
properties of processes. They are proved in~\cite{Part1}, and will be
exploited later to assess the separating power of the logic.


We start by introducing two measures on terms, that represent two ways of
defining the \emph{depth} of a process. The first definition exploits
the notion of eta normal form (see Lemma~\ref{l:eta_nf}):

\begin{defi}[Sequentiality degree, $\ds$]\label{defds}
  The sequentiality degree of a term $P$ is defined as follows:
  \begin{enumerate}[$\bullet$]
  \item $\ds(\nil)=0$, $\ds(P|Q)=\max\big(\ds(P),\ds(Q)\big)$;
  \item $\ds(\amb{n}{P})=\ds(!P)=\ds(P)$;
  \item $\ds(\capa.P)=1+\ds(P)$;
  \item $\ds(\msg n)=1$;
  \item $\ds(\abs x P)=\ds(P')+1$ where $\abs x P'$ is the eta normal form of 
    $\abs x P$.
  \end{enumerate}
\end{defi}

Intuitively, the sequentiality degree counts the number of `parcels of
interaction' (capabilities, messages, input prefixes) in a term. We
now define the \emph{depth degree}, that is sensitive to the number of
nested ambients. This quantity will be soon used in the interpretation of
some formulas of AL,  but also to define an inductive order on processes (see Subsection~\ref{subsec:soundcomplete}).

\begin{defi}[Depth degree]\label{def:depthdegree}
  The depth degree of a process is computed using a function $\dd$
  from $\MA$ processes to natural numbers, inductively defined by:
  \begin{enumerate}[$\bullet$]
  \item $\dd(\nil)~\defiDS~ 0$, $\dd(\capa.P)~\defiDS~ 0$;
  \item $\dd(\abs x P)~\defiDS~0$, $\dd(\msg n)~\defiDS~0$;
  \item $ \dd(\amb{n}{P})~\defiDS~\dd(P)+1 $;
  \item $\dd((!)P_1|\dots|(!)P_r)~\defiDS~\max_{1\leq i\leq r}
    \dd(P_i)$.
  \end{enumerate}
\end{defi}

We introduce formulas that express some kind of
\emph{possibility modalities} corresponding to the movement capabilities
and input prefix of \MA.


\begin{lem}\label{l:capaform}
  For any $\capa$, there exists a formula context $\Fmeta{\capa}.\fillhole{\cdot}$
  such that for any closed process $P$, and any formula $\A$,
$$
P\sat \Fmeta{\capa}.\fillhole{\A} \qquad\mbox{iff}\qquad \exists P',P''.~P\equiv
\capa.P'\, ,\,P'\Rcap P''~\mbox{and}~ P''\sat \A\,.
$$
For all $n$, there is a formula $\msg n$ such that
$$
P\sat\msg n  \qquad\mbox{iff}\qquad P\equiv\msg{n}\,.
$$
For all $n$, there exists a formula context $\Fmeta{? n}.\fillhole{\cdot}$
such that for all process $P$ and formula $\A$,
$$
P\sat \Fmeta{?n}.\fillhole{\A} \qquad\mbox{iff}\qquad \exists
x,P',P''.~P\equiv 
(x)P'\, ,\,(x)P'|\msg n\,\Rar\, P''~\mbox{and}~ P''\sat \A\,.
$$
\end{lem}


We will also need the \emph{necessity modalities}, that have a dual
interpretation w.r.t.  the above formulas:

\begin{lem}
  For all $\capa$, there is a formula context
  $\Ftame{\capa}.\fillhole{\cdot}$ 
  such that for all process $P$ and formula $\AAA$,
$$
P\sat \Ftame{\capa}.\fillhole{\A} \qquad\mbox{iff}\qquad \exists P'.~P\equiv
\capa.P'~\mbox{and}~ \forall P''.\,P'\Rcap P''~\mbox{implies}~ P''\sat
\AAA\,.
$$
For all $n$, there is a formula context $\Ftame{? n}.\fillhole{\cdot}$
such that, for all process $P$ and formula $\AAA$,  
$$
P\sat \Ftame{?n}.\fillhole{\AAA} \qquad\mbox{iff}\qquad \exists
P',x.~P\equiv 
(x)P'~\mbox{and}~\forall P''.(x)P'|\msg n\,\Rar \,P''~\mbox{implies}~
P''\sat \AAA\,. 
$$
\end{lem}

Each operator of the syntax of \MA{}
(Table~\ref{ta:syn}) has  thus a counterpart in the logic, except
replication. It is possible to express in AL a restricted form of
replication on formulas, by defining a formula $!\AAA$, expressing
that there are infinitely many processes in parallel satisfying
$\AAA$, modulo some additional condition on $\AAA$.
%
%
%
%
More precisely, based on Definitions~\ref{defds} and
~\ref{def:depthdegree} above, we say that a formula $\A$ is
\emph{sequentially selective} (resp.  \emph{depth selective}) if all
processes satisfying $\A$ have the same sequentiality degree (resp.
depth degree). 

\begin{lem}\label{l:repcapa}
  For all $\capa$, there exists a formula context
  $\repl{\capa}{\cdot}$ such that for all process $P$ 
  and for all sequentially selective formula $\A$, whose models are
  only of the form $\capa.R$,
  $$
  P\,\sat\, \repl{\capa}{\A} \quad\mbox{iff}\quad \exists
  P_1,\ldots,P_r.~P\equiv
  !P_1|(!)P_2|\ldots|(!)P_r~\mbox{and},\,P_i\,\sat\,\A\,, i=1\ldots
  r\,.
$$
For all $n$, there is a formula $!\msg {n}$ such that
$$
P\,\sat\,!\msg n  \quad\text{iff}\quad P\equiv!\msg{n}
\enspace.
$$
There exists a formula context $\repl{input}{\cdot}$ such that 
for all process $P$ and for all formula $\A$ sequentially selective 
whose models are only of the form $(x)P$,
$$
P\,\sat\, \repl{input}{\A} \quad\mbox{iff}\quad \exists 
P_1,\ldots,P_r.~P\equiv
!P_1|(!)P_2|\ldots|(!)P_r~\mbox{and},\,P_i\,\sat\, \A\,,
i=1\ldots r\,.
$$
\end{lem}

Similar results hold for the replicated version of the \emph{dual}
modalities.
The notion of depth selectiveness allows us to derive formulas that
capture replicated ambients:

\begin{lem}
  For all $n$, there is a formula context $!n[\fillhole{\cdot}]$ such that for
  all process $P$ and for all depth selective formula $\A$,
  $$
  P\,\sat\, !n[\fillhole{\A}] \quad\mbox{iff}\quad \exists
  P_1,\ldots,P_r.~P\equiv
  !P_1|(!)P_2|\ldots|(!)P_r~\mbox{and},\,P_i\,\sat\, n[\A]\,,
  i=1\ldots r\,.
$$
\end{lem}

By putting together these expressiveness results, we can derive
formulas characterising the equivalence class of a process w.r.t.
logical equivalence for a subcalculus of \MA, defined as follows:

\begin{defi}[Subcalculus \MAIF]\label{def:MAIF}
  Consider a process $P$, and a name $n\not\in\fn{P}$. We say that $P$
  is \emph{image-finite} if any subterm of $P$ of the form $\capa.P'$
  (resp. $(x)P'$) is such that the set
$$ \{
  P''~:~P'~\Rcap~P''\}_{/\bisMOD}$$ (resp. $\{ P''~:~P'\sub n
  x~\Rar~P''\}_{/\bisMOD}$) is finite.
%
  \MAIF{} is the set of image-finite MA processes.
\end{defi}

In the standard definition of image-finiteness, as used, e.g., to
establish inductively completeness of the Hennessy-Milner logic, one
requires that the set of outcomes of \emph{the process} is finite.
While exploring the possible outcomes (and in absence of restriction
in the process calculus), we may expose at top-level any
subterm of the process, and hence we implicitly require that all of
its subterms are image-finite in the standard sense. On the other
hand, in our case, we do not impose that $P$ has only finitely many
outcomes, but only do so for \emph{some} subterms. As a consequence,
our notion is less restrictive, and any image-finite process in the
standard sense belongs to \MAIF.

\begin{lem}[Characteristic formulas on \MAIF]\label{l:charforMAIF}
  For any closed \MAIF{} process $P$, there exists a formula $\AAA_P$
  s.t.  for any $Q$, these three conditions are equivalent:
  \begin{enumerate}[\em(1)]
  \item
  $Q\,\sat\,\AAA_P$;
  \item 
  $P\,\eqL\,Q$;
  \item $P\, \bisMOD \, Q$. 
\end{enumerate}
\end{lem}



\medskip

A final  expressiveness result that will be needed later is
the ability to test free name occurrences in a process.
\begin{lem}
  For any name $n$, there exists a formula $\copyright n$ such that
  for any $P$, $P\,\sat\,\copyright n$ iff $n\in\fn{P}$.
\end{lem}


\subsection{Soundness, and Completeness for Finite
Processes}\label{subsec:soundcomplete} 

We now study soundness and completeness of $\bisMOD$ with respect to
$\eqL$.  Soundness means that 
$\bisMOD\,\subseteq \,\eqL$, and completeness is the converse. We show
here soundness  on the whole calculus. By contrast, we only prove
completeness 
on the finite processes, deferring the general result to the next
section.
 We
chose to do this for the sake of clarity: the proof in the finite case
is much simpler, and exposes the basic ideas of the argument in the
full calculus.

\subsubsection{Soundness on full public \MA}
\label{s:coind}

In order to prove soundness (on the whole calculus), we use the
definition of $\bisE$ and the congruence property to establish that
bisimilar processes satisfy the same formulas.

\begin{thm}[Soundness of $\bisMOD$]
\label{t:soundnessLogBis}
Assume $P,Q\in\MA$, and suppose $P \bisMOD Q$. Then, for all $\A$, it
holds that $P \sat \A$ iff $Q \sat \A$.
\end{thm} 

\proof
By induction on the size of $\A$. 

\begin{enumerate}[$\bullet$]
\item
$\A = \true$.

 Nothing to prove.

\item $\A =  \myneg \B $ or  $\A =   \B_1 \orr \B_2$.

By induction and  the definition of satisfaction.

\item  
 $\A =  {\zero } $.

By definition of satisfaction and clause \reff{c:nil} of  the definition
of $\bisMOD$.

\item $\A =  {\amb n \B} $.

Then $P \equiv \amb n {P'}$ and $\satDS{P'}\B
$.  Hence $Q \equiv \amb n {Q'}$ for some $Q' \bisMOD P'$. 
By induction, $\satDS{Q'}\B$; we can therefore conclude 
that also $\satDS Q{\amb n \B}$ holds. 
\item $\A =  {\A_1|\A_2} $.

Then $P \equiv P_1|P_2$ and $\satDS{P_i}{\A_i}$.  
By clause \reff{cMOD:par} of Definition~\ref{d:bisMOD},
 $Q \equiv {Q_1|Q_2}$ for some $Q_i \bisMOD P_i$. 
By induction, $\satDS{Q_i}{\A_i}$; we can therefore conclude 
that also $\satDS Q{\A_1|\A_2}$ holds.

\item $\A = \all x \B $.

By definition of satisfaction,
$ \satDS{P}{\B \sub n x}$ for all $n$. 
The result for $Q$ then follows by induction, for $ {\B \sub n x}$ is
strictly small than  $\all x \B$.

\item $\A = \diamond \B$.

By definition of satisfaction,
there is $P'$ such that $ P \Longrightarrow  P' $  and $ \satDS{P'}\B$.
Using clause \reff{c:tau}  of  the definition
of $\bisMOD$, there is $Q'$ such that $ Q
\Longrightarrow  Q'  \bisMOD P'$. By induction, $\satDS{Q'}{\B}$; hence 
  $\satDS{Q}{\A}$.

\item $\A  = \at \B n  $ or $\A  = \A_1 \limp \A_2   $.

Follows using induction and the congruence of $\bisMOD$.\qed
\end{enumerate}




\subsubsection{Completeness, on finite
  processes}\label{s:completenessfinite}

The proof of completeness we develop here is based on the construction
of a sequence of approximants of $\bisE$, which is a standard approach
for image-finite calculi.  This works in the finite case (finiteness
implies image-finiteness), but not in presence of replication. 
The proof is however interesting on its own, and gives a
much simpler account on how the logic expresses the clauses of
$\bisMOD$ than the proof for the whole calculus.


Note that the definability of characteristic formulas
for $\bisMOD$ on \MAIF{} (see Definition~\ref{def:MAIF} and
Lemma~\ref{l:charforMAIF}) implies completeness: for two \MAIF{}
processes $P$ and $Q$, $P\,\eqL\,Q$ entails $P\,\bisMOD\,Q$. Since
\MAIF{} contains the set of finite processes, this already gives
completeness on finite processes. We nevertheless present here a proof
that is specific to the finite case,  to prepare
the ground for completeness on full public MA.
The route we are interested in for the completeness proof uses $i$-th
approximants $\bisEn i$ of relation $\bisE$, and the fact that
$\bisEomega ~\eqdef~ \bigcap_i \bisEn i$ coincides with $\bisE$.

\begin{defi}
\label{d:bisEwithEqualApprox}
We define the relations $\bisEn i$ between processes, for all $i \geq
0$, as follows.  

$\bisEn 0$ is the universal relation, and $\bisEn {i+1}$ is defined by
saying that $P \bisEn {i+1}Q $ holds if we have:
\begin{enumerate}[(1)]
\item
If $P \equiv P_1 | P_2$
then there are $Q_s$ ($s=1,2$) such that 
$Q \equiv  Q_1 |Q_2 $ and for all $s$ 
$P_s \bisEn i Q_s$.

\item 
\label{cca:in}
If $P\equiv \capa. P'$ then 
there are $Q',Q''$ such that
\begin{enumerate}
\item $Q \equiv \capa. Q'$, 
\item $Q'\Rcap Q''$, and
\item $P' \bisEn i Q'$.
\end{enumerate}
 
\item If $P \equiv \msg n$ then 
$Q \equiv \msg n $.

\item
\label{cca:abs} If $P \equiv \abs x {P'}$
then
there is $Q'$ such that
\begin{enumerate}
\item
 $Q \equiv \abs x {Q'}$ and 
\item for all $n$
there is $Q''$ such that
$ \msg n | Q \Longrightarrow Q'' $ and $P' \sub n x  \bisEn i Q''$.
\end{enumerate}
 
\item If $P \equiv \amb n {P'}$ then there is $Q'$ such that $Q 
\equiv \amb n {Q'}$ and $P' \bisEn i Q'$.

\end{enumerate} 

We set $\bisEomega ~\eqdef~ \bigcap_{i\geq 0} \bisEn i$.
\end{defi} 

\begin{lem}
\label{l:same}
$\bisEomega$ coincides with $\bisE$ on finite processes.
\end{lem} 
\begin{proof}
Standard approximation result  (finite processes are image finite).
\end{proof}

\begin{lem}
\label{l:bisEomega}
Let $P,Q$ be two \emph{finite} processes.  If $P \,\eqL\, Q$ then $P
\bisEomega Q$.
\end{lem} 

\proof
Suppose $P \bisEomegaNOT Q$. Then there is $i$ such that
$P \bisEnNOT i Q$. We prove, by induction on $i$, that in this 
case we can find a formula $\A$ such that $\satDS P \A $ holds but $\satDS 
Q\A $ does not. 

For $i = 0$, this trivially holds since the hypothesis $P\bisEnNOT 0Q $ is absurd 
for  $\bisEn 0$ being the universal
relation.

Now the case $i+1$, for $i\geq 0$.
We proceed by case analysis:

\begin{enumerate}[(1)]
\item
$P \equiv  P_1 |  P_2$,
and  for all 
 $Q_1, Q_2$  such that 
$Q \equiv  Q_1   | Q_2$ there is $t$  ($1 \leq t \leq 2$)
such that $P_t \bisEnNOT i Q_t$.

Modulo $\equiv$, there is a finite number, say $s$,  of pairs of
processes 
$ Q_1 ,  Q_2$ such that 
$Q \equiv  Q_1 | Q_2$ (note that by hypothesis $P$ is finite). 
Call $Q_{t,u}$ the $t$-th process of the $u$-th pair.  Then for
all $u$ ($1\leq u \leq s$) there is $t$  such that 
$P_t \bisEnNOT i Q_{t,u}$. By induction, there is $\A_{t,u}$ such that 
\[ \satDS{P_t}{\A_{t,u}}  \qquad\mbox{ and  }\qquad  \satNOT{Q_{t,u}}{\A_{t,u}}\,.
\]
Define
\[B_t \eqdef \bigwedge_{u.\,1\leq u \leq s {\mathrm{and}~P_t \bisEnNOT i Q_{t,u}} } \A_{t,u}\,. \]
Then 
\[\satDS P {B_1  | B_2}\,, \]
\noindent whereas
\[\satNOT Q {B_1  | B_2}\,. \]

\item 
$P\equiv \capa . P'$; then necessarily  
 $Q \equiv \capa. Q'$, and for all $Q_t$ such that
  $Q' \Rcap {Q_t} $,
it holds  that 
$P' \bisEnNOT i Q_t$.

By induction, for all $t$ there is $\A_t$ such that $\satDS
{P'}{\A_t}$ but $\satNOT{Q_t}{\A_t}$. Since $Q$ is finite, there is
only a finite number of such processes $Q_t$ (up to $\equiv$). Write
$(Q_t)_{t\in I}$ for  this set of processes up to $\equiv$ (we pick a
representant for each $\equiv$-equivalence class), and call $\A_t$ the
formula corresponding to each $Q_t$.  Define
\[\A \eqdef  \Fmeta{\capa} .  \fillhole{\bigwedge_{t\in I} {\A_t}}\,,\]
\noindent using the standard notation for the (finite)
conjunction of the $\A_t$s. Then $\satDS P \A$ but $\satNOT Q \A$.

\item $P\equiv\msg n$, and $Q\not\equiv\msg n$: then $P\sat\msg n$, and 
$Q\not\sat\msg n$.

\item
  $P \equiv \abs x {P'}$,
 $Q \equiv \abs x {Q'}$ and 
there is  $n$
such that for all  $Q_t$ such that
$ \msg  n | Q \Longrightarrow Q_t $, it holds that
 $P''  \bisEnNOT i Q_t$, for $P'' \eqdef  P' \sub n x$.

Modulo $\equiv$, there is only a finite number of such $Q_t$s, say $Q_1, 
\ldots, Q_s$. By induction, there are formulas 
$\A_1 , \ldots, \A_s$ with $\satDS {P''}{\A_t}$ and  $\satNOT {Q_t}{\A_t}$. 
We introduce as above the notation $(Q_t)_{t\in I}$, and we define
\[\A \eqdef  \Fmeta{?n}.\fillhole{\bigwedge_{t\in I} \A_t}\,. \]
Then $\satDS P \A$, but $\satNOT Q \A$, because 
whenever 
$\msg n | Q \Longrightarrow Q_t$, it holds that  
$\satNOT{Q_t}{\A_t}$.

\item  $P \equiv \amb n {P'}$,
$Q \equiv \amb n {Q'}$ and $P' \bisEnNOT i Q'$.

By induction there is $\A'$ with $\satDS {P'}{\A'}$ but  $\satNOT
{Q'}{\A'}$.
 Define $\A \eqdef \amb n{\A'}$; then $\satDS P \A $ but $\satNOT Q\A$.\qed

\end{enumerate}

\begin{thm}[Completeness on finite processes]
\label{t:bisEomega}
Let $P,Q$ be two finite closed processes.  If $P \,\eqL\, Q$ then $P
\bisMOD Q$.
\end{thm}
\begin{proof}
Follows from Lemma~\ref{l:same} and \ref{l:bisEomega}.
\end{proof}


\section{Completeness of \texorpdfstring{$\bisMOD$}{bisMOD} in the 
full calculus}\label{s:inductive}
\newcommand{\imgset}[1]{\mathcal{E}_{#1}^{\Downarrow}}
\newcommand{\frset}[1]{\mathcal{E}_{#1}^{\textrm{frz},N}}
The proof we have presented in the finite case cannot be used directly
in the full \MA{} calculus, because we lack the image-finiteness
hypothesis, which allowed us to show that the limit $\bisEomega$
coincides with $\bisE$. In this section, we present a proof of the
completeness of $\bisMOD$ for all processes. To do this, we establish
the existence, for any processes $P,Q$, of a formula $\F_{P,Q}$ such
that $P\sat \F_{P,Q}$, and such that $Q\sat \F_{P,Q}$ holds if and only
if $P\bisMOD Q$. This result is hence weaker than the existence of
characteristic formulas, but it does not require
image finiteness.

We sketch the structure of the proof. Our approach exploits two
technical devices, that we introduce first. We start by proving some
lemmas related to the sequentiality degree of a term
(Definition~\ref{defds}), which allows us to define a sound induction
principle on \MA{} processes. This principle supports the introduction
of an inductive characterisation of $\bisMOD$. The second technical
device we introduce is the set of \emph{frozen subterms} of a process,
that intuitively corresponds to the collection of subterms appearing
under guards (capabilities or input prefixes) in a given term. These
two technical notions are then used to define \emph{local
  characteristic formulas}, which correspond to a relaxed notion of
characteristic formula w.r.t.  logical equivalence.  An important fact
about the set of frozen subterms of a process is that it enjoys a kind
of subject reduction property; this allows us to replace the
potentially infinite set of images of a term with a finite set when
constructing local characteristic formulas. 

\subsection{An inductive characterisation of
  \texorpdfstring{$\bisMOD$}{bisMOD}}\label{subsec:bigind} 

We now establish some properties related to the sequentiality degree
of processes. These allow us to introduce a well-founded order on
terms which supports the definition of an inductive relation that
coincides with $\bisMOD$. 


\begin{lem}\label{dsred}
  Let $P,Q$ be two terms of \MA{}. Then:
\begin{enumerate}[(1)]
\item if $P\equiv Q$, then $\ds(P)=\ds(Q)$;
\item
  if $P\,\rar\, Q$ or  $P\,\arr \mu\, Q$ then $\ds(P) \geq \ds(Q)$.
\end{enumerate}
\end{lem}
\begin{proof}
  1 is immediate, as is the result on $\arr\mu$ in 2. For $P\,\rar\,
  Q$, we reason by induction on the height of the derivation of
  $P\,\rar\, Q$.
\end{proof}
\begin{cor}
For all $\capa$, if $P\Rcap Q$, then $\ds(P)\geq \ds(Q)$.
\end{cor}

This result will be important for the justification of
Definition~\ref{d:bigind} below.

\begin{lem}
  \label{l:sdc}
  For any closed process $P\in\MA$, there exists a formula
  $\F_{\ds(P)}$ such that:
\begin{enumerate}[$\bullet$]
\item $P~\sat~\F_{\ds(P)}$, and
\item for any term $Q$, if $Q~\sat~\F_{\ds(P)}$, then
  $\ds(Q)\geq\ds(P)$.
\end{enumerate}
\end{lem}

\proof
  We can assume that $P$ is eta normalised.  Let us first reason by
  induction on $\ds(P)$:
  \begin{enumerate}[$\bullet$]
  \item for $\ds(P)=0$, $\F_{\ds(P)}=\ltrue$ is sufficient.
  \item for $\ds(P)>0$, let us assume that there exist formulas
    $\F_{\ds(P')}$ for any $P'$ such that $\ds(P')<\ds(P)$.  We reason
    by induction on $P$.
    \begin{enumerate}[$-$]
    \item the case $P=\nil$ is impossible.
    \item for $P=P_1|P_2$, there is $i$ such that $\ds(P)=\ds(P_i)$.
      Then we may choose $\F_{\ds(P)}=\F_{\ds(P_i)}|\ltrue$. In the
      same way, let us set $\F_{\ds(\msg n)}=\F_{\msg n}$,
      $\F_{\ds(!P)}=\F_{\ds(P)}|\ltrue$ and
      $\F_{\ds(\amb{n}{P})}=\amb{n}{\F_{\ds(P)}}$.
    \item for $P=\capa.P'$, we use the general induction
      hypothesis to construct $\F_{\ds(P')}$. Let us then take
      $\F_{\ds(P)}=\Fmeta{$\capa$}.\F_{\ds(P')}$. Then
      $P~\sat~\F_{\ds(P)}$,    and for any $Q$ such that
      $Q~\sat~\F_{\ds(P)}$, we deduce (from Lemma~\ref{l:capaform})
      that there are  $Q',Q''$ such that $Q\equiv\capa.Q'$ and
      $Q'\Rcap Q''$ with $Q''~\sat~\F_{\ds(P')}$. Now by
      Lemma~\ref{dsred}, $\ds(Q)-1=\ds(Q')\geq\ds(Q'')$, and by
      induction hypothesis $\ds(Q'')\geq\ds(P')=\ds(P)-1$, so that
      finally $\ds(Q)\geq\ds(P)$.
    \item for $P=\abs x P'$, we use the general induction hypothesis
      to get $\F_{\ds(P')}$. Let us then take $\F_{\ds(P)}= \exists
      x.~\Fmeta{?x}.\F_{\ds(P')}$. Then $P~\sat~\F_{\ds(P)}$,
      and for any $Q$ such that $Q~\sat~\F_{\ds(P)}$, we deduce (from
      Lemma~\ref{l:capaform}) that there are $n,Q',Q''$ such that
      $Q\equiv\abs x Q'$ and $Q_1=\msg{n}|(x)Q'\Rar Q''$ with
      $Q''~\sat~\F_{\ds(P')}$. Now by Lemma~\ref{dsred},
      $\ds(Q_1)-1=\ds(Q)-1=\ds(Q'\sub n x)\geq\ds(Q'')$, and by
      induction hypothesis $\ds(Q'')\geq\ds(P')=\ds(P)-1$, so that
      finally $\ds(Q)\geq\ds(P)$.\qed
    \end{enumerate}
  \end{enumerate}

\noindent A similar result can be proved for the depth degree of a process:

\begin{lem}
  \label{l:ddc}
  For any closed process $P\in\MA$, there exists a formula
  $\F_{\dd(P)}$ such that:
\begin{enumerate}[$\bullet$]
\item $P~\sat~\F_{\dd(P)}$, and
\item for any term $Q$, if $Q~\sat~\F_{\dd(P)}$, then
  $\dd(Q)\geq\dd(P)$.
\end{enumerate}
\end{lem}

\begin{proof}
  We reason  as in the proof of the previous lemma.
\end{proof}

\begin{cor}\label{c:dsdd}
If $P\,\bisMOD\, Q$, then $\ds (P)=\ds (Q)$ and $\dd(P)=\dd(Q)$.
\end{cor}
\begin{proof}
  By Theorem~\ref{t:soundnessLogBis}, $P\,\bisMOD\,Q$ implies
  $P\,\eqL\, Q$, which gives the result.
\end{proof}

The sequentiality degree can be used as a basis for inductive
reasoning on processes up to reductions of some subterms. This is
formalized by the following definition: 

\begin{defi}[Well-founded order]
  Given two processes $P$ and $Q$ , we write $P<Q$ (or $Q>P$)  if either
  $\ds(P)<\ds(Q)$ or $P$ is a strict subterm of $Q$. 
\end{defi}

\begin{lem}\label{l:wellfounded}~
\begin{enumerate}[$\bullet$]
\item $<$ is well-founded.
\item Suppose $P$ is of the form either $\capa.P'$ or $(x)P'$, 
and suppose moreover $P>Q$ and $Q\Rcap Q'$ for some $\capa$. Then $P>Q'$.
\end{enumerate}
\end{lem}

\proof
\begin{enumerate}[$\bullet$]
\item Well-foundedness:  if $P$ is a strict subterm of
  $Q$, then $\ds(P)\leq \ds(Q)$.
\item $P>Q'$: follows from Lemma~\ref{dsred}.\qed
\end{enumerate}\medskip

In order to give an inductive characterisation of $\bisMOD$, we
establish the following results about $\bisMOD$. These are
\emph{inversion properties}, in the sense that they allow one to
deduce, from $P\,\bisMOD\,Q$, with $P$ having a given shape,
consequences about the shape of $Q$.




\begin{lem}[Inversion results for $\bisMOD$]\label{l:inversion}
  Let $P, P_1, P_2, Q$ be processes of \MA.  Then
\begin{enumerate}[\em(1)]
\item $\nil~\bisMOD Q$ iff $Q\equiv \nil$.
\item $\amb{n}{P}~\bisMOD~Q$ iff there exists $Q'$ such that
  $Q~\equiv~\amb{n}{Q'}$ and $P~\bisMOD~Q'$.
\item $P_1|P_2~\bisMOD~Q$ iff there exist $Q_1,Q_2$ such that
  $Q~\equiv~Q_1|Q_2$ and $P_i~\bisMOD~Q_i$ for $i=1,2$.
\item $!P~\bisMOD~Q$ iff there exist $r\geq 1, s \geq r, Q_i$ 
  ($1\leq i\leq s$) such that 
  $Q \equiv \Pi_{1\leq i \leq r} {!Q_i}\,|\, \Pi_{r+1\leq i \leq s}
  {Q_i}$, and $P~\bisMOD~Q_i$ for $i=1\dots s$.
\item $\capa.P~\bisMOD~Q$ iff there exists $Q'$ such that
  $Q~\equiv~\capa.Q'$ with $P~\Rcap\bisMOD~Q'$ and
  $Q'~\Rcap\bisMOD~P$.
\item $\msg n~\bisMOD~Q$ iff $Q\equiv \msg n$
\item $\abs x P~\bisMOD~Q$ iff there exists $Q',m$. such that 
$m\not\in\fn{P}\cup\fn{Q}$, $Q\equiv\abs x 
Q'$  $Q|\msg m\Rar \bisMOD P\sub m x$ and $\abs x P |\msg m \Rar \bisMOD Q'\sub m x$.
\end{enumerate}
\end{lem}
\begin{proof}
  We first leave out the fourth case.
  
  For the other cases, the left to right implications follow by the
  fact that, in each case, the corresponding clauses in the
  definitions of $\bisE$ and $\bisMOD$ are almost the same.

  For the right to left implication,
  cases 1 and 6 hold by reflexivity of $\bisMOD$, and cases 2 and 3
  follow from congruence of $\bisMOD$ (Corollary~\ref{c:bisEcong}).
  Case 5 is similar to the corresponding condition in $\bisE$ (note
  that all other conditions are trivially fulfilled).
  
  We explain case 7 in more details. We take $P,Q,Q',x,m$ satisfying
  the required properties, and further introduce processes $P_1$ and
  $Q_1$ by imposing $(x)P|\msg m\Rar P_1\bisMOD Q'\sub m x$ and
  $Q|\msg m\Rar Q_1\bisMOD P\sub m x$. To show that $P\bisMOD Q$, we
  need to show that these processes satisfy the condition for
  receptions in the definition of $\bisE$
  (Definition~\ref{d:bisEwithEqual}), all other requirements being
  satisfied.
  Consider an arbitrary name $m'$, we want to show that there exist
  $P'',Q''$ such that $Q|\msg {m'}\Rar Q''$, $P\sub {m'} x\bisMOD
  Q''$, $(x)P|\msg {m'} \Rar P''$ and $P'' \bisMOD Q'\sub {m'} x$. By
  hypothesis, this holds for $m'=m$, by taking $P''=P_1$ and
  $Q''=Q_1$. Otherwise, we set $P''=P_1\sub {m'} m$ and
  $Q''=Q_1\sub {m'} m$.  Then $(x)P|\msg {m'}=((x)P|\msg m)\sub {m'} m
  \Rar P_1\sub {m'} m=P''$ since $\Rar$ is closed under name
  replacement, and $Q|\msg {m'} \Rar Q''$ for the same reason.
  Moreover, since $\bisMOD$ is also closed under name replacement
  (Lemma~\ref{l:subst_bis}), we deduce from the hypothesis $P_1\bisMOD
  Q'\sub m x$ that $P''\bisMOD Q'\sub {m'} x$, and similarly from
  $Q_1\bisMOD P\sub m x$ that $Q''\bisMOD P\sub {m'} x$. As a
  consequence, the condition is established for all $m'$.  Note that
  the hypothesis about $m$ being fresh for $P,Q$ is crucial in the
  proof above.
  
  \smallskip
  
  We are thus left with case 4. The right to left implication holds
  because, if we define $\mathcal R$ as $\bisE$ extended with all
  pairs of the form $(P, \Pi_{1\leq i \leq r} {!Q_i}\,|\, \Pi_{r+1\leq
    i \leq s} {Q_i})$, with the above conditions, then $\mathcal R$
  satisfies the clauses of $\bisE$, hence $\mathcal R$ is $\bisE$.
  We now consider the left to right implication. First, note that by
  applying clauses~\ref{cMOD:par} and~\ref{c:nil} of
  Def.~\ref{d:bisMOD}, it can be shown that for any two bisimilar
  processes $P,Q$, if $P\equiv P'|P'|..P'|P''$, where $P$ contains at
  least $n$ copies of some single process $P'$, then necessarily
  $Q\equiv Q_1|..|Q_n|Q'$ with $Q_i\bisE P'$ for all $i$. This entails
  the left to right implication in the case where $P$ is a single
  process. When $P$ is not single, we write $P\equiv \Pi_{1\leq i \leq
    r} {!P_i}\,|\, \Pi_{r+1\leq i \leq s} {P_i}$, where $P_1,..,P_s$
  are single processes. Thanks to the congruence rule
  $!(R_1|R_2)=!R_1|!R_2$, $!P\equiv !P_1|..|!P_s$. Assume $!P\bisE Q$.
  Applying the inversion rule for parallel composition, we have
  $Q\equiv Q_1|Q_s$ with, for every $i$, $!P_i\bisE Q_i$, that is,
  using our reasoning on single processes, $Q_i\equiv\Pi_{1\leq j \leq
    r_i} {!Q_{i,j}}\,|\, \Pi_{r_i+1\leq j \leq s_i} {Q_{i,j}}$.  Using
  the law $!R\equiv !R|!R$, it is possible to choose all $r_i$ equal,
  and similarly applying $!R\equiv!R|R$ we can choose all $s_i$ equal.
  It is then a matter of rearranging the $Q_{i,j}$ in $Q_1'|..|Q_s'$
  to write $Q$ in the expected form.  
\end{proof}



We can now define the inductively defined relation that characterises
$\bisMOD$.

\begin{defi}
\label{d:bigind}
Let $\bigind$ be the binary relation $P\bigind Q$ defined by induction
on $P$ for the order $<$ as follows:
\begin{enumerate}[(1)]
\item $\nil~\bigind~Q$ if $Q\equiv \nil$.
\item $\amb{n}{P}~\bigind~Q$ if there exists $Q'$ such that
  $Q~\equiv~\amb{n}{Q'}$ and $P~\bigind~Q'$.
\item $P_1|P_2~\bigind~Q$ if there exist $Q_1,Q_2$ such that
  $Q~\equiv~Q_1|Q_2$ and $P_i~\bigind~Q_i$ for $i=1,2$.
\item $!P~\bigind~Q$ if there exist $r\geq 1, s \geq r, Q_i$ 
  ($1\leq i\leq s$) such that 
  $Q \equiv \Pi_{1\leq i \leq r} {!Q_i}\,|\, \Pi_{r+1\leq i \leq s}
  {Q_i}$, and $P~\bigind~Q_i$ for $i=1\dots s$.
\item $\capa.P~\bigind~Q$ if there exists $Q'$ such that
  $Q~\equiv~\capa.Q'$ with $P~\Rcap\bigind~Q'$ and
  $Q'~\Rcap\bigind~P$.
\item $\msg n~\bigind~Q$ if $Q\equiv \msg n$
\item $\abs x P~\bigind~Q$ if there exists $Q',m$. such that
  $m\not\in\fn{P}\cup\fn{Q}$, $Q\equiv\abs x Q'$, $Q|\msg m\Rar
  \bigind P\sub m x$ and $\abs x P |\msg m \Rar \bigind Q'\sub m x$.
\end{enumerate}
\end{defi}


\begin{thm}
  Relation $\bigind$ is well defined. Moreover, relations $\bigind$
  and $\bisMOD$ coincide.
\end{thm}
\begin{proof}
  The definition of $\bigind$ is justified using
  Lemma~\ref{l:wellfounded}.
    The inclusion $\bisMOD\subseteq\bigind$ is established using the
  results of Lemma~\ref{l:inversion}, which correspond precisely to
  the defining clauses of $\bigind$. The converse inclusion follows from Lemma~\ref{l:inversion} too.
\end{proof}

\subsection{Frozen subterms}

We now introduce the notion of frozen subterms of a process.  The
frozen subterms of a process correspond to occurrences that do not
participate in immediate interactions but that may play a role in
future reductions.

In the reminder, we  use $N$ to range over sets of names. Unless
otherwise stated, we  always implicitly suppose that such a set is
finite.


\begin{defi}[Frozen subterms]  \label{d:froz}  ~
 Let $N$ be a set of names; the set $\frozen{P}$ is
  defined by induction on $P$ as follows:
\begin{enumerate}[$\bullet$]
  \item $\frozen{\nil}=\frozen{\msg n}= \emptyset$;
  \item $\frozen{P_1|P_2}=\frozen{P_1}\cup\frozen{P_2}$;
  \item $\frozen{!P}=\frozen{P}$;
  \item $\frozen{\capa.P}=\{P\}\cup\frozen{P}$;
  \item $\frozen{(x)P}=\bigcup_{n\in N}\{ P\sub n x\}\cup \frozen{P\sub n x}$.
\end{enumerate} 
  \end{defi}

If $P, P'$ are two structurally congruent terms, then, modulo
$\equiv$, $\frozen{P}=\frozen{P'}$.  Hence this set (in its quotiented
version with respect to $\equiv$) is uniquely determined by the
structural congruence class of $P$.

\begin{lem}[Finiteness of $\frozen{P}$]\label{l:finite}
  For any $P\in\MA$, if $N$ is finite, then the set obtained by taking
  the quotient of $\frozen{P}$ w.r.t. $\equiv$ is finite.
\end{lem}

\begin{proof}
By induction on $P$.
\end{proof}

Not only is $\frozen{P}$ finite, but, as expressed by the following
result, this set is preserved by reduction, in the following sense:

\begin{lem}
  \label{l:predictive}
  Let $P,Q$ be two processes such that $P\,\rar\, Q$ or $P\rcap Q$ for
  some $\capa$, and assume $\fn{P}\subseteq N$.  Then the quotient of
  $\frozen{Q}$ w.r.t. $\equiv$ is included in the quotient of
  $\frozen{P}$ w.r.t. $\equiv$.
\end{lem}

\begin{proof}
  We recall that relation $\rcap$ is defined on the syntax of
  processes (see Definition~\ref{d:statt}), and the result follows by
  definition of $\frozen{P}, \frozen{Q}$.
  
  For $\rar$, we reason by induction on the derivation of $P\,\rar\,
  Q$. The cases corresponding to movement transitions follow from
  $\rcap$.  So the only way a reduction could alter the set of frozen
  terms is through name substitutions generated by communications, and
  this is handled by the condition $\fn{P}\subseteq N$.
\end{proof}

\subsection{Local characteristic formulas and completeness}
\label{subsec:local:charac}

The purpose of this subsection is to derive local characteristic
formulas, defined as follows:

\begin{defi}[Local characteristic formula]\label{d:restricted} Let
  $\mathcal{E}$ be a set of terms, $P$ a term and $\F$ a formula. We
  say that $\F$ is a characteristic formula for $P$ on $\mathcal{E}$
  (or, alternatively, a $\mathcal{E}$-characteristic formula for $P$)
  if
  \begin{enumerate}[$\bullet$]
  \item $P\,\sat\, \F$, and
  \item for any $Q\in\mathcal{E}$, if $Q\,\sat\, \F$ then $Q\,\bisMOD\,
    P$.
  \end{enumerate}
\end{defi}

Note that the converse of the second condition always holds, due to
soundness of $\bisMOD$ (Theorem~\ref{t:soundnessLogBis}): if
$Q\in\mathcal{E}$ and $Q\,\bisMOD\, P$, then $Q\,\sat\, \F$.

With this definition, completeness of $\bisMOD$ boils down to the
existence, for any processes $P, Q$, of a characteristic formula of
$P$ on the set  $\{ Q\}$.  
Although we do not define directly such a formula, this idea guides
the construction of the completeness proof. More precisely, we 
reason inductively on the sequentiality degree of processes, and
manipulate two sets of terms, given a process $P$:
\begin{enumerate}[$\bullet$]
\item $\imgset{P}~\defiDS~\{P',~\exists~\capa.\,~P~\Rcap~P'\}$, that
  collects the possible evolutions of $P$,
\item and $\frset{P}~\defiDS~\{P',\frozen{P'}\subseteq\frozen{P}\}$,
  that intuitively is the set of processes whose possible evolutions
  can be captured using the evolutions of $P$.
\end{enumerate}

We want to establish the existence, for all $P, Q$, of a local
characteristic formula for $P$ on $\imgset Q$ and $\frset Q$.  We
first prove the following result:

\begin{lem}\label{l:fr2img}
  If a formula $\mathcal{F}$ characterises $P$ on $\frset Q$ and
  $N\supseteq\fn Q$, then $\mathcal{F}$ characterises $P$ on $\imgset
  Q$.
\end{lem}
\begin{proof}
  Follows from Lemma~\ref{l:predictive}.
\end{proof}

The following lemma describes the construction of a local
characteristic formula for guarded terms (of the form $\capa.P$ or
$(x)P$) on $\frset Q$, provided we can compute, given several
(smaller) processes $R$, local characteristic formulas on $\imgset R$:

\begin{lem}\label{l:img2fr}
  Consider two processes $P$ and $Q$, and a set $N$ of names such that
  $\fn P\cup\fn Q\,\subseteq\,N$. Assume moreover that, for all
  $Q'\in\frozen Q$, we can construct a formula $\F_{P,Q'}$
  characterising $P$ on $\imgset {Q'}$ and a formula $\F_{Q',P}$
  characterising $Q'$ on $\imgset P$.  We then have:
\begin{enumerate}[$\bullet$]
\item for all $\capa$ there exists a formula
characterising $\capa.P$  on $\frset Q$,
\item for all $n$ such that $P$ is not of the form $\msg n|(y)P'$ with
  $n\not\in\fn{P'}$, and for all $x$ with $x\not\in\fv{P}$,
  there exists a formula characterising $(x)\big(P\sub x n\big)$ on
  $\frset Q$.
\end{enumerate}
\end{lem}\medskip

\proof~
\begin{enumerate}[$\bullet$]
\item Let $\capa$ be a given capability. Set $\mathcal{E}=\{
  Q'\in\frozen Q: \forall P' ~\mbox{s.t.}~ P\Rcap P',\, P'\not\bisMOD
  Q'\}$; $\mathcal{E}\subseteq \frozen Q$, so by Lemma~\ref{l:finite},
  $\mathcal{E}$ is finite, and we can define the formula:
$$
F~\eqdef~\Fmeta\capa\fillhole{\bigwedge_{Q'\in\frozen Q}\F_{P,Q'}}~\land~
~\Ftame\capa\fillhole{\bigwedge_{Q'\in\mathcal{E}} \non \F_{Q',P}}\,. 
$$
We prove first that $\capa.P\,\sat\, F$; by hypothesis, $P\,\sat\,
\F_{P,Q'}$ for all $Q'\in\frozen Q$, so that we have 
$\capa.P\,\sat\,\Fmeta\capa\fillhole{\bigwedge_{Q'\in\frozen
    Q}\F_{P,Q'}}$. Let $P'$ 
be such that $P\Rcap P'$, and consider any $Q'\in \mathcal{E}$.  Then
by hypothesis $P'\,\sat\, \F_{Q',P}$ would imply $P'\bisMOD Q'$, and hence
$Q'\not\in \mathcal{E}$, which is contradictory. So
$P'\,\sat\,\bigwedge_{Q'\in\mathcal{E}} \non \F_{Q',P}$, and finally $P\,\sat\,
F$. 

Conversely, consider $R\in \frset Q$ such that $R\,\sat\, F$. We show that
$R\bisMOD P$. First, there is $Q'$ such that $R\equiv\capa.Q'$ and
$Q'\Rcap\,\sat\, \F_{P,Q''}$ for all $Q''\in\frozen Q$. Since $R\in\frset
Q$ and $Q'\in\frozen R$, $Q'\in\frozen Q$, so $Q'\Rcap\,\sat\, \F_{P,Q'}$,
and by hypothesis, $Q'\Rar\bisMOD P$, which gives the first part of
the condition to have $\capa.P\bigind R$ (Definition~\ref{d:bigind}).
Furthermore, since $R$ satisfies the `necessity' part of the formula
$F$, $Q'\,\sat\, \bigwedge_{Q''\in\mathcal{E}} \non \F_{Q'',P}$, that is
$Q'\not\in\mathcal{E}$. Thus, there is $P'$ with $P\Rcap P'$ and
$P'\bisMOD Q'$, which gives the second part of the condition.

\item Let $n,x$ be chosen as in the statement of the lemma. We set
  $P_0=(x)\big(P\sub x n\big)$.  Similarly as before, we define
  $\mathcal{E}=\{ Q'\in\frozen Q: \forall P'~\mbox{s.t.}~ P\Rar P',\,
  P'\not\bisMOD Q'\}$; again $\mathcal{E}\subseteq \frozen Q$, so
  $\mathcal{E}$ is finite, and we may define the formula:
$$
\begin{array}{rl}
F~\eqdef& ~ \non\copyright n
\\ &\land~\Fmeta{?n}\big(\mathsf{NonEta}\land
\bigwedge_{Q'\in\frozen Q}\F_{P,Q'}\big)\qquad~
\\ &\land~ \Ftame{?n}\big(\mathsf{NonEta}~ \rar~\bigwedge_{Q'\in\mathcal{E}} \non \F_{Q',P}\big)
\\
\multicolumn{2}{l}{\mathrm{with}} 
\\
~\qquad\mathsf{NonEta} \eqdef & \non\big(\msg n |(\non\copyright n \land 
\Fmeta{?n}\ltrue)\big)
\end{array}
$$
Intuitively, the role of formula $\mathsf{NonEta}$ is to detect
when the reducts of a process satisfying $F$ stop being eta-equivalent
to the initial state.

Let us prove that $P_0\,\sat\, F$: $n\not\in\fn{P_0}$ by construction,
$P_0|\msg n\Rar P$, $P\,\sat\,\mathsf{NonEta}$ and $P\,\sat\, \bigwedge
\F_{P,Q'}$ by hypothesis, so $P_0$ satisfies the second conjunct in
$F$.  Take $P'$ such that $P_0|\msg n\Rar P'$ and
$P'\,\sat\,\mathsf{NonEta}$; we prove that $P'\,\not \sat\, \F_{Q',P}$ for all
$Q'\in\mathcal{E}$.  Since $P'\,\sat\,\mathsf{NonEta}$, $P_0|\msg
n\not\equiv P'$, so $P\Rar P'$.  As a consequence, $P'\,\sat\, \F_{Q',P}$
iff $P'\bisMOD Q'$. Then by definition of $\mathcal{E}$,
$P'\,\sat\,\bigwedge_{Q'\in\mathcal{E}} \non \F_{Q',P}$. As this holds for
all $P'$, we have that $P_0\,\sat\, F$.

Let us now prove that if $R\in\frset Q$ and $R\,\sat\, F$, then $P\bisMOD
R$. Consider such a process $R$.  Then $n\not\in\fn{R}$, and there
exists $Q',R'$ such that $R\equiv (x)Q'$ and $R|\msg n\Rar R'$ with
$R'\,\sat\,\mathsf{NonEta}\land\bigwedge_{Q'\in\frozen Q}\F_{P,Q'}$. Let
$(x)Q''$ be the head eta normal form of $(x)Q'$.  By definition,
$Q''\sub n x$ belongs to $\frozen Q$, and any reduction $(x)Q'|\msg
n\Rar T$ where $T$ is not eta equivalent to $(x)Q'|\msg n$ goes
through the state $Q''\sub n x$ (i.e., that reduction can be written
$(x)Q'|\msg n\Rar Q''\sub n x\Rar T$).  Due to the definition of
$\mathsf{NonEta}$, we actually have that $R'\not \equiv_E (x)Q'|\msg
n$, so $Q''\sub n x\Rar R'$. Since $R'\,\sat\, \F_{P,Q''\sub n x}$,
$R'\bisMOD P$ and the first part of the condition for input in
Definition~\ref{d:bigind} is satisfied.  Moreover, $R|\msg n\Rar
Q''\sub n x$ and $Q''\sub n x\,\sat\,\mathsf{NonEta}$, so $Q''\sub n
x\,\sat\,\bigwedge_{Q'\in\mathcal{E}} \non \F_{Q',P}$. Since $Q''\sub
n x\in\frozen Q$, we finally have $Q''\sub n x\not\in\mathcal{E}$, that
is there is $P'$ such that $P\Rar P'$ and $P'\bisMOD Q''\sub n x$.
This proves the second condition for $P_0\bigind (x)Q''$, and since
$(x)Q''\equiv_E R$, we finally have $P_0\bisMOD R$.\qed
\end{enumerate}

We now prove that given $P$, we can deduce a local characteristic
formula for $P$ from local characteristic formulas for its guarded
subterms.

\begin{lem}\label{l:fr2fr}
  Consider two processes $P$ and $Q$, and a set of names $N$, and
  suppose that, for each subterm of $P$ of the form $\capa.P'$ or
  $(x)P'$, we can construct a $\frset Q$-characteristic formula. Then
  there exists a $\frset Q$-characteristic formula for $P$.
\end{lem}

\begin{proof}
  We assume, without loss of generality, that all occurrences of the
  replication operator in $P$ are immediately above a guarded process
  (this is always possible up to $\equiv$).
  
  We construct such a formula $\F_P$ by induction on $P$. The cases for
  $\nil$, parallel composition, and ambient are easy. Formulas for
  messages and replicated messages have been given above, and by
  hypothesis, we have formulas for guarded processes. We are thus left
  with the case of replicated terms.

  If $P=!n[P]$, then $\F_P=!n[\fillhole{\F_P'}]$ is a $\frset
  Q$-characteristic 
  formula, since $\F_{P'}$ is depth selective (all processes satysfying
  $\F_{P'}$ are intensionally bisimilar to $P'$, so their depth degree
  is equal to $\dd{(P')}$ -- see Corollary~\ref{c:dsdd}).  If
  $P=!\capa.P'$, then $\F_P=\repl{\capa}{\F_{\capa.P'}}$, since
  $\F_{\capa.P'}$ is sequentially selective. We reason in the same way
  for the case $P=!(x)P'$.

\end{proof}

\begin{lem}\label{l:keylemma}
  For all $P,Q$ and $N\supseteq\fn{P}\cup\fn{Q}$, there exist
  characteristic formulas for $P$ on $\imgset Q$ and $\frset Q$.
\end{lem}

\begin{proof}
  From Lemma~\ref{l:fr2img}, it is sufficient to construct a local
  characteristic formula on $\frset Q$. We remark that without loss of
  generality, $P,Q$ can be choosed so that every binding $(x)P$ involves a different variable, and 
  this is enough to build characteristic formulas for the set $
  N$ enriched with distinct names $n_x$
  associated to all variables $x$ occurring in $P$ and $Q$.  We reason
  by induction on $\ds (P)$. If $\ds (P)= 0$, then $P$ has no guarded
  subterms, and the conditions of Lemma~\ref{l:fr2fr} are fullfilled,
  which implies the existence of a local characteristic formula for
  $P$.
  
  Assume now $\ds (P)>0$, and, for all $P'$ such that $\ds (P')<\ds
  (P)$, and for all $Q$, there exists a characteristic formula for
  $P'$ on $\frset Q$. Consider a process $Q$. By Lemma~\ref{l:fr2fr},
  the existence of a $\frset Q$-characteristic formula for $P$ can be
  proved by establishing the existence of a $\frset Q$-characteristic
  formula for each guarded subterm of $P$ of the form $\capa.P'$ or
  $(x)P'$. Consider such a guarded subterm $\capa.P'$.  We have $\ds
  (P')<\ds (P)$, so by induction there exists a formula $\F_{P,Q'}$
  which is a $\imgset{Q'}$-characteristic formula for $P'$ for each
  $Q'\in\frozen Q$. Moreover, by induction, we also have a formula
  $\F_{Q',P'}$ which is a characteristic formula for $Q'$ on $\imgset
  {P'}$ when $\ds (Q')\leq\ds P')<\ds (P)$. In the case $\ds (Q')>\ds
  (P')$, we define $\F_{Q',P}$ as the formula $\F_{\ds{(Q')}}$ given in
  Lemma~\ref{l:sdc}. This formula characterises $Q'$ on $\imgset{P'}$:
  $Q'\sat \F_{Q',P}$ by Lemma~\ref{l:sdc}, and if $P''\in\imgset{P'}$
  then $\ds{(P'')}\leq\ds{(P')}<\ds{(Q')}$, so $P''\not\sat
  \F_{\ds{(Q')}}$.  Hence the requirements of Lemma~\ref{l:img2fr} are
  fullfilled, and there exists a $\frset Q$-caracteristic formula for
  $\capa.P'$.
  
  Similarly, consider a subterm of the form $(x)P'$, and write
  $(x)P''$ for its eta normal form.  As above, we have local
  characteristic formulas $\F_{P''\sub {n_x} x,Q'}$ and $\F_{Q',P''\sub
    {n_x} x}$ by induction and using Lemma~\ref{l:sdc} with a similar
  reasoning. Since $(x)P''$ is in normal form, all requirements of
  Lemma~\ref{l:img2fr} are satisfied, so that there exists a $\frset
  Q$-characteristic formula for $(x)P''$, which is also a
  characteristic formula for $(x)P'$ by Lemma~\ref{c:equiv-bis}.
  
  Finally, we have characteristic formulas for all guarded subterms,
  and by Lemma~\ref{l:fr2fr}, we have a $\frset Q$-characteristic
  formula for $P$.
\end{proof}

\begin{thm}[Completeness of $\bisMOD$]\label{t:completeness}
  In \MA, $\eqL \,\subseteq~ \bisMOD$.
\end{thm}

\begin{proof}
  Let $P,Q$ be two terms such that $P\,\not\bisMOD\,Q$.
  By Lemma~\ref{l:keylemma}, there is a formula $F$ characterising $P$ on 
  $\imgset Q$. We have
  $P~\sat~F$.
  We then have $Q~\in\imgset Q$, and $Q~\sat~F$ implies
  $P~\bisMOD~Q$. Hence, since by hypothesis $P\,\not\bisMOD\,Q$,
  $Q\,\not\!\!\sat\,F$, and $P\, \not\!\!\eqL\,Q$.
\end{proof}

\begin{cor}
\label{c:chara}
In \MA, relations $\eqL$, $\bisMOD$ and $\bigind$ coincide.
\end{cor}


\section{Characterizations of logical equivalences}\label{s:axiom}


In this section, we compare logical equivalence and standard
equivalence relations on processes, like behavioural equivalence and
structural congruence. We give an axiomatization of $\eqL$ on
\MAIFsyn, a subcalculus of \MA{} in which image-finiteness is
guaranteed by a syntactical condition (Definition~\ref{d:MAIFsyn}
below). We shall see that AL is very intensional, in the sense that
$\eqL$ is `almost equal' to $\equiv$. More precisely, we show that
logical equivalence coincides with $\equiv_E$, the relation obtained
by extending structural congruence with the eta law
(Definition~\ref{d:etalaw}).
We establish the following chain of (dis)equalities, on \MAIFsyn:
$$
\equiv~~\subsetneq~~\equiv_{E}~~=~~\eqL~~=~~\bisMOD~~\subsetneq~~\approx\,. 
$$
We then move to the study of a variant of \MAIFsyn{} in which
communication is synchronous, and show that logical equivalence
coincides with $\equiv$ on this calculus. We end this section with a detailed discussion of the treatment of name restricition.

\subsection{Extensionality and intensionality}\label{subsec:extint}

We use the characterisation of $\eqL$ as $\bisMOD$ to compare logical
equivalence with  barbed congruence ($\approx$) and  structural
equivalence ($\equiv$). 
%
We start by studying the difference between
$\eqL$ and $\wbc$.

\subsubsection{Non-extensionality}

\begin{thm}
\label{t:eqLbis}
Relation $\eqL$ is strictly included in $\wbc$.
\end{thm} 

\proof
  The inclusion follows from $\eqL \subseteq\, \bisMOD$ and $\bisMOD\,
  \subseteq\, \wbc $ (the second inclusion is essentially a
  consequence of the congruence of $\bisMOD$).

 The strictness of the inclusion is proved by the following laws, that
 are valid for $\wbc$ but not for $\bisMOD$:

\begin{enumerate}[(1)]

\item $\ina n . \ina n  =   \ina n | \ina n$

\item $\abs x {\abs y \nil} = \abs  x \nil | \abs y \nil$



\item $\abs x {\msg x} = \nil$.\qed
\end{enumerate}\medskip
 
The third axiom is typical for behavioural  equivalences in
 calculi where communication is  asynchronous.
 The first equality can be derived from a more general law, called the
 \emph{distribution law} in~\cite{hirschkoff:pous:fossacs07}: $M.(P~|~
 M.P|\dots|M.P) = M.P~|~M.P|\dots|M.P$ (where $M$ appears the same
 number of times on both sides of the equality). A similar law is
 valid for the input prefix, from which the second equality above is
 derived as an instance.
Probably the above are  not  the only laws  that make $\eqL$
   finer than $\wbc$, but a complete axiomatization of $\wbc$ over $\eqL$ is out of the scope of this paper.
%

\subsubsection{Intensionality}

We now provide a precise account of the difference between $\eqL$ and
$\equiv$, in the setting of the subcalculus \MAIFsyn, defined as
 below. We recall that a process is finite if it does not use the
replication operator.

\begin{defi}[\MAIFsyn]\label{d:MAIFsyn}
The subcalculus \MAIFsyn is defined by the grammar:
$$
P~::=~~\nil~\big|~P|P~\big|~!P~\big|~n[P]~\big|~\capa.P_0~\big|~\msg
n~\big|~(x)P_0
$$
where $P_0$ is a finite process.
\end{defi}
In \MAIFsyn, we impose finiteness after any form
  of interaction; in contrast,  processes exhibiting an `infinite spatial
  structure', such as $!a[b[\nil]]$ are allowed.


\begin{lem}
  All processes of \MAIFsyn{} are image-finite.
\end{lem}
\begin{proof}
  \MAIFsyn{} is included in \MAIF{} since the finiteness condition on
  $P_0$ in Definition~\ref{d:MAIFsyn} implies that $ \{
  P'~:~P_0~\Rcap~P'\}_{/\bisMOD}$ and $\{ P'~:~P_0\sub n
  x~\Rar~P'\}_{/\bisMOD}$ respectively are finite sets. Any process in
  \MAIFsyn{} is thus in \MAIF{}, and is hence image-finite in the
  sense of Definition~\ref{def:MAIF}.
\end{proof}


$\MAIFsyn$ strictly contains the finite calculus we considered for the
completeness proof in Section~\ref{s:completenessfinite}. Therefore,
Theorem~\ref{t:bisEomega} does not apply, but Corollary~\ref{c:chara},
which holds for the whole calculus, does.
As \MAIF, 
\MAIFsyn{} is image-finite, in the sense of Definition~\ref{def:MAIF}.
While in the former subcalculus this property is guaranteed at a
semantical level, in \MAIFsyn{} it follows from a syntactic
restriction (we forbid replication in process $P_0$ -- see
Definition~\ref{d:MAIFsyn}).
  
We will see in Section~\ref{s:undecidability} that $\MAIFsyn$ is
Turing complete.



\smallskip

  We let \emph{normalised structural congruence}, written $\equivE$,
  be the relation defined by the rules of $\equiv$ plus the eta law
  (see Definition~\ref{d:etalaw}).

\begin{lem}
\label{l:etacorrect} 
$\equiv_E\,\subseteq\,\bisE$.
\end{lem}

\begin{proof}
  It is enough to prove that given $P,Q$ such that $P\rar_{\eta} Q$,
  we have $P\bisMOD Q$.  We reason by induction on $P$, following
  Lemma~\ref{l:inversion}.  In that lemma, the situations
  corresponding to the operators of parallel composition, ambients and
  capability prefixes are easy because of commutation properties of
  $\rar_{\eta}$. In the cases of $\nil$ and of messages, there is no
  redex for $\rar_{\eta}$.
  
  So we only have to examine the clause for the input condition in
  $\bisMOD$.  Let $n$ be a fresh name and write $P\equiv\abs x P'$,
  $Q\equiv\abs x Q'$. We have to prove that $P|\msg n~\Rar\bisMOD~
  Q'\sub n x$ and $Q|\msg n~\Rar\bisMOD~P'\sub n x$.  The reduction
  $P\rar_{\eta}\; Q$ can follow from two reasons: either $P\equiv\abs
  x\big(\msg x|\abs x Q'\big)$, or $P'\rar_{\eta}\; Q'$. In the first
  case, the proof is straightforward, and in the second case, the
  induction hypothesis allows us to conclude.
\end{proof}

The converse of this lemma is the difficult part of the
characterisation of $\eqL$ in \MAIFsyn. This is proved by showing that
two intensionally bisimilar finite processes have essentially the same number
of prefixes and messages. Using the separative power given by the
logic, this entails that $\bisE\subseteq\equiv_E$ on \MAIFsyn. It has
to be stressed that we rely here on the syntactical finiteness condition
defining \MAIFsyn, and that our approach does not apply to, e.g.,
\MAIF.

We write $\OPmess R $ for the number of messages in $R$, and $\OP R $
for the number of capabilities and abstractions in $R$.

\begin{lem}
\label{l:mess_pref_RED}
Let $P,Q$ be two finite processes.
Suppose $P \longrightarrow P'$. Then 
\begin{enumerate}[\em(1)]

\item
$\OPmess P \geq \OPmess {P'} $;

\item  $\OP P \geq \OP {P'} $.
\end{enumerate} 
\end{lem} 

\begin{proof}
By induction on the derivation of $P\rar P'$.
\end{proof}

\begin{lem}
\label{l:numOpmess}
Let $P,Q$ be two finite processes.  Suppose that $P \bisE Q$, and that
both $P$ and $Q$ are eta-normalised. Then $\OPmess P = \OPmess Q$.
\end{lem}   

\proof
Suppose $\OPmess P >  \OPmess Q$. We prove that we derive a
contradiction.   We proceed by a case analysis on the shape of
$P$ (i.e., the number of its operators)

\begin{enumerate}[$\bullet$]
\item $P = P_1 | P_2$. Then, by definition of $\bisE$, it must be $Q
  \equiv Q_1 | Q_2$ with $P_i \bisE Q_i$. Now, for some $i$, we should
  have $\OPmess{P_i} \neq \OPmess{Q_i}$, which is impossible, by the
  induction on the shape.

\item $P = \capa . P'$. Then, by definition of $\bisE$, it must be 
$Q \equiv  \capa . Q'$  and $Q' \Rcap Q'' \bisE P'$. 
 It will then be, by Lemma~\ref{l:mess_pref_RED}(1),  
 $\OPmess{P} = \OPmess{P'} > \OPmess{Q''}$, which is impossible, by the
induction on the shape.

\item $P = \abs x {P'}$. Then, by definition of $\bisE$, it must be 
$Q \equiv \abs x {Q'}$; moreover, for $n $ fresh, 
there must be $Q''$ such that $\msg n | \abs x {Q'} \Longrightarrow
Q'' \bisE P' \sub n x$. 

If the reduction $\msg n | \abs x {Q'} \Longrightarrow Q''$ contains
at least one step, then we would have $\OPmess{P'\sub n x} = \OPmess P
> \OPmess{Q'} \geq \OPmess {Q''}$ and therefore, by induction on the
shape, we could not have $Q'' \bisE P' \sub n x$.

Therefore, suppose  $ Q'' = \msg n | \abs x {Q'}$. 
Then   $Q'' \bisE P' \sub n x$ implies 
$P' \sub n x \equiv \msg n | \abs x {P''}$, for some $\abs x {P''}$
with $n$ fresh for $P$ and $Q$. 
Hence, since $n$ was chosen fresh, the original process $P$ must have
been  of
the form $\abs x {(\msg x | \abs  x {P''} )}$. This means that,
modulo $\equiv$,  $P$
was not eta-normalised, thus contradicting an hypothesis of the
lemma.

\item If $P = \msg n$ then by definition of $\bisE $ we should have $Q
  \equiv \msg n$, which is impossible, since the hypothesis is
  $\OPmess P > \OPmess Q $.\qed
\end{enumerate}\medskip

\begin{lem}
\label{l:numOpBIS}
Let $P,Q$ be two finite processes. 
Suppose $P \bisE Q$, and that  both $P$ and $Q$ are eta-normalised. Then 
$\OP P =  \OP Q$.
\end{lem}   

\proof
Suppose $\OP P >  \OP Q$. We prove that we derive a
contradiction.  We proceed by induction on 
the shape of $P$. 

\begin{enumerate}[$\bullet$]

\item If $P = \nil$ then  $Q \equiv \nil$.

\item $P = P_1 | P_2$. Then, by definition of $\bisE$, it must be 
$Q \equiv Q_1 | Q_2$ with $P_i \bisE Q_i$. Now, for some $i$, we
should have $\OP{P_i} \neq \OP{Q_i}$, which is impossible, by the
induction on the shape.  

\item $P = \capa . P'$. Then, by definition of $\bisE$, it must be 
$Q \equiv  \capa . Q'$  and $Q' \Rcap Q'' \bisE P'$. 
Then
\[\OP{P'} = \OP P -1 > \OP Q -1 = \OP{Q'} \geq \OP{Q''} \]
Hence $\OP{P'} > \OP{Q''}$,
 which is impossible by the
induction on the shape. 

\item $P = \abs x {P'}$. Then, by definition of $\bisE$, it must be 
$Q \equiv \abs x {Q'}$; moreover, given $n $ fresh, 
there must be $Q''$ such that $\msg n | \abs x {Q'} \Longrightarrow
Q'' \bisE P' \sub n x$. 

Moreover, by the previous lemma we know that $\OPmess P = \OPmess Q$,
and we should also have  $\OPmess {P' \sub n x} = \OPmess {Q''}$

The reduction  $\msg n | \abs x {Q'} \Longrightarrow
Q''$ must contain at least one step, for otherwise 
we could not have $\OPmess {P' \sub n x} = \OPmess {Q''}$. 
For the same reason, during these reductions only the message $\msg n$ 
may have been consumed (no other messages).  
Thus $\msg n | \abs x {Q'} \Longrightarrow Q''$ can be written 
as 
\[ \msg n  | \abs x {Q'} \longrightarrow Q' \sub n x \Longrightarrow
Q''\,, \]
where $\OP{Q'} = \OP{Q' \sub n x}$
 and also $ \geq \OP{Q''}$ 
(Lemma~\ref{l:mess_pref_RED}(2)).

Therefore we have $\OP{P'\sub n x} = \OP P - 1 > \OP{Q} - 1 = \OP{Q'}
\geq \OP{Q''}$. By the induction on the shape, this is in
contradiction with $Q'' \bisE P' \sub n x$.\qed
\end{enumerate}\medskip

\begin{lem}
\label{l:strong}
Let $P,Q$ be two finite processes. 
Suppose $P \bisE Q$, with both $P$ and $Q$ eta-normalised. If $P \arr\mu P'$, 
then there is $Q' $ such that $Q \arr\mu Q' \bisE P'$.
Similarly,  if $P \longrightarrow  P'$, 
then there is $Q' $ such that $Q \longrightarrow  Q' \bisE P'$.
\end{lem} 

\begin{proof}
From Lemmas~\ref{l:numOpBIS} and \ref{l:numOpmess}:  if $Q$ performed
more than one action, then it would consume one more  prefix or message than 
$P$.
\end{proof}

\begin{thm}
\label{t:bisSTR}
Let $P,Q$ be processes of $\MAIFsyn$. 
Suppose  $P \bisE Q$, with both $P$ and $Q$ eta-normalised. Then $P \equiv 
Q$. 
\end{thm} 

\proof
By induction on the   shape of
$P$. 


\begin{enumerate}[$\bullet$]

\item If $P = \nil$ then also $Q \equiv \nil$.
\item Suppose  $P = P_1 | P_2$.  Then, by definition of $\bisE$, $Q \equiv Q_1 
  | Q_2$ with $P_i \bisE Q_i$. By induction, $P_i \equiv Q_i$. Hence
  also $P \equiv  Q$.

\item Suppose $P=!P'$. Then, by Lemma~\ref{l:inversion}, there are $r$
  and some $(Q_i)_{1\leq i\leq r}$ such that $$Q\equiv
  !Q_1|(!)Q_2|\dots|(!)Q_r , $$ and $P'\bisE Q_i$ for all $i$. By
  induction, $P'\equiv Q_i$ for all $i$, so finally
  $Q\equiv!Q_1\equiv P$.

\item $P = \capa . P'$. By definition of $\bisE$,
 $Q \equiv \capa . Q'$ and there is $Q''$ such that $Q' \Rcap
 Q'' \bisE P'$. By construction of $\MAIFsyn$, $P',Q'$ are finite, so
 that we may apply Lemma~\ref{l:strong}. Then it must be $Q' = Q''$,
 and therefore by induction $Q' \equiv P'$. We conclude that $P \equiv Q$.

\item $P = \msg n, \amb n {P'}$: straightforward.

\item $ P = \abs  x {P'}$. By definition of $\bisE$, we have 
$Q \equiv \abs x {Q'}$, and again by construction of $\MAIFsyn$, $P',Q'$ 
are finite. Since $\bisE$ is a congruence,  given $n$, 
$\msg n | P \bisE \msg n | \abs x {Q'}$.  We have 
$\msg n | P  \longrightarrow P' \sub n x$, hence by
Lemma~\ref{l:strong}, $ \msg n | \abs x {Q' }\longrightarrow Q' \sub n
x \bisE  P' \sub n x$.  By induction, $P' \sub n x \equiv Q' \sub n x$;
since this holds for any $n$, $P' \equiv Q'$.\qed 
\end{enumerate} 

\begin{cor}
\label{c:bisSTR}
Let $P,Q$ be processes of \MAIFsyn. Then  $P\, \eqL\, Q$ iff  $P \equivE Q$.
\end{cor}
 
\begin{proof}
First, $\eqL\subseteq\,\bisMOD$ by Theorem~\ref{t:bisEomega}, and
$\bisMOD\,\subseteq\,\equivE$ by Theorem~\ref{t:bisSTR}. Conversely,
$\equivE\,\subseteq\,\bisMOD$ by Lemma~\ref{l:etacorrect}, and
$\bisMOD\,\subseteq\eqL$ by Theorem~\ref{t:soundnessLogBis}.
\end{proof}


\subsection{Synchronous
communications}\label{s:synchronous} 

We now consider a variant of Mobile Ambients where communication is
\emph{synchronous}. For this the production $\msg \eta $ for messages
in the grammar of MA in Table~\ref{ta:syn} is replaced by the
production $\msg \eta . P$.  Communication is thus synchronous: in
$\msg \eta . P$, the process $P$ is blocked until the message $\msg
\eta$ has been consumed.  Reduction rule \trans{Red-Com} becomes:
\[
\shortaxiomC { \msg n .Q | \abs x P \longrightarrow Q| P \sub n x }
{Red-Com}
\]
In the remainder of this subsection, terms belonging to the
synchronous version of the calculus will be referred to simply as
`processes'. Since our goal here is to study how the result given by
Corollary~\ref{c:bisSTR} changes when moving to a synchronous
calculus, we focus directly on \MAIFss, the set of all terms of the
synchronous calculus in which processes guarded by prefixes are finite
(along the lines of Definition~\ref{d:MAIFsyn} that introduces
\MAIFsyn). We shall see that in \MAIFss, the eta law fails and the
equivalence relation induced by the logic is precisely structural
congruence.




In order to show this, we have to port the results about
(asynchronous) \MA{} to the synchronous case.  The co-inductive
characterisation in terms of $\bisMOD$ (that is,
Theorems~\ref{t:soundnessLogBis} and~\ref{t:bisEomega}) remains true,
provided that in the definition of intensional bisimulation the communication
clauses are replaced by the following:

\begin{enumerate}[$\bullet$]
\item 
  If $P \arr{!n} {P'}$, then there is $Q'$ such that 
  $Q   \Ar{! n}  {Q'}$ and $P' \RR Q'$. 
\item
 If $P   \arr{? n}  {P'}$
then
there is $Q'$ such that $Q   \Ar{? n}  {Q'}$ and $P' \RR Q'$.
\end{enumerate}

Accordingly, we have to change the definition of syntactical
intensional bisimulation by adapting
the following clauses for communicating processes:
\begin{enumerate}[$\bullet$]
\item
  
  If $P \equiv\abs x{P'}$ then there is $Q'$ such that $Q \equiv \abs
  x{Q'}$ and for all $n$ there is $Q''$ such that $Q' \sub n x
  \Longrightarrow Q'' $ and $P' \sub n x \RR Q''$.
  
\item If $P \equiv \msg n. {P'}$ then there is $Q'$ such that $Q
  \equiv \msg n . {Q'}$ and $Q' \Longrightarrow Q'' \RR P'$.
 \end{enumerate} 
 
 As shown in~\cite{Part1}, formulas similar to those that are needed
 in the asynchronous case can be derived for the synchronous calculus.
 In particular, we have:\medskip


\begin{lem}[\cite{Part1}]~
\begin{enumerate}[$\bullet$]
\item For all $\A$, there is a formula $\Fmeta {?n}.\fillhole{\A}$ such that
for all $P$, $P\sat\Fmeta{?n}.\fillhole{\A}$ iff there is $P'$ such that
$P\equiv\abs x P'$ and $P'\sub n x\Rar\sat\A$.
\item For all $\A$, there is a formula $\Fmeta {!n}.\fillhole{\A}$ such that
for all $P$, $P\sat\Fmeta{!n}.\fillhole{\A}$ iff there is $P'$ such that
$P\equiv\msg n. P'$ and $P'\Rar\sat\A$.
\end{enumerate}
\end{lem}


Using this result, the soundness and completeness proofs for $\bisMOD$
with respect to $\eqL$ follow exactly the same scheme as in the
asynchronous case (see Sections~\ref{s:char} and~\ref{s:inductive}),
except that we do not need to reason on eta-normalised terms.

\begin{thm}[Soundness and completeness of $\bisMOD$]
  Given two processes $P$ and $Q$ of synchronous Mobile Ambients,
  $P\,\bisMOD\,Q$ iff $P\,\eqL\,Q$.
\end{thm}

We now derive the counterpart of the properties we have established
above for \MAIFsyn{} about the number of messages and prefixes in a
term.

\begin{lem}
  Suppose $P \longrightarrow P'$, where $P$ is a finite process.
  Then
\begin{enumerate}[\em(1)]
\item
$\OPmess P \geq \OPmess {P'} $;
\item  $\OP P \geq \OP {P'} $.
\end{enumerate} 
\end{lem} 

\begin{proof}
By induction on the derivation of $P\rar P'$.
\end{proof}

\begin{lem}
  Let $P,Q$ be two finite processes and suppose $P \bisE Q$.
  Then $\OPmess P = \OPmess Q$.
\end{lem}   

\proof
Suppose $\OPmess P >  \OPmess Q$. We prove that we derive a
contradiction.   We proceed by a case analysis on the shape of
$P$ (ie, the number of its operators)

\begin{enumerate}[$\bullet$]
\item $P = P_1 | P_2$. Then, by definition of $\bisE$, it must be 
$Q \equiv Q_1 | Q_2$ with $P_i \bisE Q_i$. Now, for some $i$, we
should have $\OPmess{P_i} \neq \OPmess{Q_i}$, which is impossible, by the
induction on the shape.

\item $P = \capa . P'$. Then, by definition of $\bisE$, it must be 
$Q \equiv  \capa . Q'$  and $Q' \Rcap Q'' \bisE P'$. 
 It will then be, by Lemma~\ref{l:mess_pref_RED}(1),  
 $\OPmess{P} = \OPmess{P'} > \OPmess{Q''}$, which is impossible, by the
induction on the shape.

\item $P = \msg n . P'$. Then $Q \equiv \msg n .Q'$ and $P' \bisE
  Q'$.  But 
 $\OPmess{P'} > \OPmess{Q'}$, which by induction is impossible.

\item $P = \abs x {P'}$.  Then $Q \equiv \abs x {Q'}$ and for all $h$
  fresh, $Q' \sub hx \bisE \Longrightarrow Q''$ and $P' \sub hx \bisE
  Q''$ and  $\OPmess{P''} > \OPmess{Q''}$, so we can conclude by induction.\qed
\end{enumerate} 

\begin{lem}
  Let $P,Q$ be two finite processes, and suppose $P \bisE Q$.
  Then $\OP P = \OP Q$.
\end{lem}   

\proof
Suppose $\OP P >  \OP Q$. We prove that we derive a
contradiction.  We proceed by induction on 
the shape of $P$. 

\begin{enumerate}[$\bullet$]

\item If $P = \nil$ then  $Q \equiv \nil$.

\item $P = P_1 | P_2$. Then, by definition of $\bisE$, it must be 
$Q \equiv Q_1 | Q_2$ with $P_i \bisE Q_i$. Now, for some $i$, it
should be $\OP{P_i} \neq \OP{Q_i}$, which is impossible, by the
induction on the shape.  

\item $P = \capa . P'$. Then, by definition of $\bisE$, it must be 
$Q \equiv  \capa . Q'$  and $Q' \Rcap Q'' \bisE P'$. 
Then
\[\OP{P'} = \OP P -1 > \OP Q -1 = \OP{Q'} \geq \OP{Q''} \]
Hence $\OP{P'} > \OP{Q''}$,
 which is impossible by the
induction on the shape. 

\item $P = \msg n . P'$. Similar to capability case.

\item $P = \abs x {P'}$.  Then $Q \equiv \abs x {Q'}$ and there is
  $Q''$ such that $Q' \sub hx \bisE \Longrightarrow Q''$ and $P' \sub hx \bisE
  Q''$. There is no consumption of messages, hence 
 $\OP{P' \sub hx} > \OP{Q''}$, and we can conclude using induction.\qed
\end{enumerate}

\begin{lem}
Let $P,Q$ be two finite processes, and suppose $P \bisE Q$. If $P
\arr\mu P'$, then there is $Q' $ such that $Q \arr\mu Q' \bisE P'$.
Similarly, if $P \longrightarrow P'$, then there is $Q' $ such that $Q
\longrightarrow Q' \bisE P'$.
\end{lem} 

\begin{proof}
  From the two previous lemmas: if $Q$ performed more than one
  action, then it would consume one more prefix or message than $P$.
\end{proof} 

\begin{thm}
  Let $P,Q$ be two processes in \MAIFss{}, and suppose $P \bisE Q$. Then $P
  \equiv Q$.
\end{thm} 

\begin{proof}
  By induction on the shape of $P$ (almost exactly as in
  Theorem~\ref{t:bisSTR}).
\end{proof}

\begin{cor}
Let $P,Q$ be processes of \MAIFss. Then  $P \eqL Q$ iff  $P \equiv Q$.
\end{cor}


\subsection{Name restriction}

\newcommand{\MAnu}{MA$^\nu$}
\newcommand{\satnu}{\ensuremath{\sat^{\nu}\,}}

In this section, we consider the variant of MA, noted here \MAnu, that
includes name restriction $(\nu n)\, P$. We discuss, among previous
results, which ones remain valid, and which ones have to be amended.

Adding name restriction involves several modifications in the
definition of the calculus and of the logic. Name $n$ is bound in
$(\nu n)\,P$, and the definition of $\fn P$ is modified accordingly.
Regarding structural congruence, we add alpha conversion for $\nu$, as
well as the following laws:
\[
\begin{array}{c}
(\nu n)\,\nil\equiv\nil
\qquad
(\nu n)(\nu m)\,P\equiv(\nu m)(\nu n)\,P
\qquad
(\nu n)\,(P|Q)\equiv P\,|\,(\nu n)\,Q\mbox{ if $n\not\in\fn P$}
\\[.5em]
(\nu n)\,m[P]\equiv m[(\nu n)\,P]
\qquad
\capa.(\nu
n)\,P\equiv(\nu n)\,\capa.P\mbox{ if $n\not\in\fn \capa$}
\end{array}
\]
The last rule is not always present in the definition of structural
congruence. It is not an essential rule, but including it makes our
some technical details simpler.


\medskip

In the logic, additional connectives are introduced,
as in~\cite{Cardelli::Gordon::NameRestriction::01}, to handle
restriction and the associated notion of freshness of names: formulas
can also be of the form $n\reveal\AAA, \AAA\hide n$, or $\jamie
n.\,\AAA$.  Accordingly, the enriched notion of satisfaction, written
\satnu, is given by:
%
%
\begin{enumerate}[$-$]
\item $P\sat^{\nu} n\reveal\AAA$ iff
  $P\equiv (\nu n)\,P'$ and $P'\satnu\AAA$ for some $P'$;
\item $P\satnu \AAA\hide n$ if $(\nu n)\,P\satnu \AAA$;
\item $P\satnu \jamie n.\,\AAA$ if there is $n'\notin(\fn P\cup\fn
  \AAA)$ such that $P\satnu\AAA\sub {n'}{n}$.
\end{enumerate}

To illustrate this new setting, we consider the two following formulas:
\begin{mathpar}
  { \mathsf{free}(n)~~\eqdef~~ \neg n\reveal\top \hspace{2cm}
    \mathsf{public}~~\eqdef~~\jamie n.\,\non \big(n\reveal\,
    \textsf{free}(n)\big) \enspace.  }
\end{mathpar}
A process $P$ satisfying $\mathsf{free}(n)$ cannot reveal $n$, which
means that $n$ necessarily occurs free in $P$. In turn, if $P$
satisfies $\mathsf{public}$, then it cannot reveal a name $n$ so as to
exhibit free occurrences of $n$, which means that $P$ is structurally
congruent to some $P'\in$ MA.

Formula $\mathsf{public}$ hence provides a way of selecting processes
belonging to \MA{} among the processes in \MAnu. We can indeed adapt
any formula $\AAA$ we have used in the paper into a formula $\AAA'$
such that whenever $P\satnu\AAA'$, then $P\equiv P'$ for some $P'$ in
MA such that $P'\sat\AAA$; in particular, formulas of the form
$\AAA_1\rtr\AAA_2$ are translated into formulas of the form
$(\BBB_1\land\mathsf{public})\rtr\BBB_2$.



\medskip

In presence of name restriction, we can adapt rather easily several
important results of the paper as follows (for each item, we indicate
the part of the paper we refer to):
\begin{enumerate}[$\bullet$]
\item a new `intensional' rule must be added to the definition of
  $\bisMOD$ (Def.~\ref{d:bisMOD}): if $P\equiv(\nu n)\,P'$, then there
  is $Q'$ such that $Q\equiv (\nu n)\, Q'$ and $P'\bisMOD Q'$;
\item with this definition, it is possible to establish a soundness
  result ($\bisMOD\subseteq\eqL$, Theorem~\ref{t:soundnessLogBis}), and
  completeness for finite processes (processes without replication,
  Theorem~\ref{t:bisEomega});
\item characteristic formulas are derivable for processes of the form
  $(\nu n_1)\dots(\nu n_k)\,P$, where $P$ is a `public' process in
  \MAIF{} (Lemma~\ref{l:charforMAIF}): we rely on name revelation to
  get rid of the topmost restrictions, and then translate the
  characteristic formula for $P$ using the approach sketched above;
\item logical equivalence coincides with structural congruence
  enriched with eta conversion for processes of the form $(\nu
  n_1)\dots(\nu n_k)\,P$, with $P$ a public process in \MAIFsyn{}
  (Corollary~\ref{c:bisSTR}).
\end{enumerate}

The difficult point, that we leave for future work, is to analyse
processes that can generate unboundedly many names, i.e., in which
restriction occurs under replication.
%
%
Characteristic formulas seem much more difficult to obtain for such
processes.  We do not know at present how to derive completeness in
absence of an image finiteness hypothesis (in particular, we do not
see how a counterpart of Lemma~\ref{l:predictive} can be obtained).



\section{(Un)decidability of logical equivalence
\label{s:undecidability}}

In this section we define the encoding of a Turing Machine in
\MAIFsyn. The purpose of this encoding is to establish that logical
equivalence in undecidable on \MAIF.

The definition of the encoding requires the introduction of some
constructions that will be given as (\MAIFsyn) contexts. To ease the
reading of our definitions, 
we shall sometimes work with \emph{parametrised contexts}, which are
context definitions that depend on some values (names, words, or
movements of the head of the Turing Machine). Additionally, some
parametrised definitions shall be written \texttt{foo}$(p);P$: here,
\texttt{foo} is the name of the definition, whereas $p$ and $P $ are
parameters ($P$ being a process); the notation emphasizes the
sequentiality between the process being introduced and $P$.

\begin{rem}
  The results in this section improve and extend a preliminary version
  presented in~\cite{Sedal}. By the time the writing of this paper was
  completed, Busi and Zavattaro \cite{busi:zavattaro:tcs:2004} have
  studied encodings of another universal machine, namely the Random
  Access Machine, into a subset of \MA. Their encodings are
  syntactically more coincise than the one below of a Turing Machine.
  However, Busi and Zavattaro make use of combinations of operators
  that are not licit in \MAIFsyn\ (i.e., their encodings are not
  encodings into \MAIFsyn). Also, while longer, the encoding of Turing
  Machines makes use of components which accomplish simple tasks and
  which interact with each other in simple manners. Correspondingly,
  each step of the proof, which follows the reductions of the encoding
  of a Turing Machine, is rather straightforward. For these reasons we
  maintain the schema of the original encoding in~\cite{Sedal}.
\end{rem}

\subsection{Ribbons}\label{sec:ribbons}

\textbf{Digits and words.} We associate to booleans true and false two
names \ttrue{} and \ffalse{}. We call these names \emph{digits}, and
range over digits with $d, d'$. A word will be the result of a
(possibly empty) concatenation of digits. The empty word shall be
written $\epsilon$. We range over words with $w, w', w_1, w_2$. Given
a word $w$ consisting in $r$ digits (with $r\geq 1$), we shall
sometimes write $w^1\dots w^r$ to refer to the digits of $w$. This
should not be confused with notation $\ffalse^n$, that we will
sometimes use to represent the word consisting in $n$ times digit
$\ffalse$ (this should be clear from the context).

We start with the definition of the support of the Turing Machine:
ribbons can be in differents states (frozen, growing, work ribbon,
old), and are defined as follows:

\newcommand{\arraytopic}[1]{\multicolumn{3}{l}{~~\mbox{\textbf{#1}}}}

\medskip

\noindent{
\begin{minipage}{\textwidth}
\noindent $  \begin{array}{lcl}
    \arraytopic{Cells and Words}\\
    \mactext{cell}(d)\hole & := & \amb{cell}{\amb{d}{\nil}~|~
      !\openamb{wo}~|~\hole~}\\
    \mactext{word}(w)\hole & := &
    \mactext{cell}(w^1)\fillhole{\mactext{cell}(w^2)\fillhole{\dots
        \mactext{cell}(w^r)\hole\dots}}\quad\quad (w = w^1 w^2\dots w^r)\\[1em]
  \end{array}
  $

\noindent $  \begin{array}{lcl}
    \arraytopic{Ribbon Extensor}\\
    \mactext{deadextcode} & := &
    !\openamb{coin}.\openamb{newcell}.\inamb{cell}.
    \amb{coin}{\nil}
    \\
    & & ~|~
    !\amb{newcell}{\mactext{cell}(\ffalse)\fillhole{\outamb{ext}}}\\
    \mactext{sendstart} & := &
    \amb{msg}{\outamb{ext}.!\outamb{cell}~|~
      \outamb{ribbon\_left}.\amb{start}{\inamb{TM}}}\\
    \mactext{ExtensorFrozen} & := & \amb{ext}{
      \mactext{deadextcode}~|~\openamb{coin}.\mactext{sendstart}}\\
    \mactext{ExtensorAlive} & := & \amb{ext}{\amb{coin}{\nil}~|~
      \mactext{deadextcode}~|~\openamb{coin}.\mactext{sendstart}}\\
    \mactext{ExtensorDead} & := & \amb{ext}{\mactext{deadextcode}}\\[1em]
  \end{array}
  $

\noindent $  \begin{array}{lcl}
    \arraytopic{Ribbons}\\
    \mactext{cleaninst} & := & \openamb{cleaner}.\openamb{runclean}~|~
    \amb{runclean}{\mactext{deadcleancode}}\\
    \mactext{deadcleancode} & := & !\openamb{\ffalse}~|~
    !\openamb{\ttrue}~|~!\openamb{cell}~|~ !\openamb{wo}\\
    \mactext{FrozenRibb}(w) & := & \amb{ribbon\_left}
    {\mactext{cleaninst}~|~
      \mactext{word}(w)\fillhole{\mactext{ExtensorFrozen}}}\\
    \mactext{GrowingRibb}(w) & := & \amb{ribbon\_left}
    {\mactext{cleaninst}~|~
      \mactext{word}(w)\fillhole{\mactext{ExtensorAlive}}}\\
    \mactext{WorkRibb}(w_1,w_2)\hole & := & \amb{ribbon\_left}
    {\mactext{cleaninst}}\\
    & &
    ~|~\mactext{word}(w_1)\fillhole{\hole{}~|~
      \mactext{word}(w_2)\fillhole{\mactext{ExtensorDead}}}\\
    \mactext{OldRibb} & := &
    \amb{ribbon\_left}{\mactext{deadcleancode}~|~
      \mactext{ExtensorDead}}
  \end{array}
  $
\end{minipage}
}

\smallskip

All names used in the definitions above are supposed to be pairwise
distinct. In particular, $TM$ is the name we shall use for the ambient
containing the Turing Machine (see Definition~\ref{def:TM}).
The ribbon is represented as a nesting of ambients named $cell$, each
of which contains an empty ambient named $d$, where $d$ is the digit
value of the cell: this corresponds to the definitions of
$\mactext{cell}(d)$ and $\mactext{word}$ -- the $!\openamb{wo}$
subterm is there to trigger the computation of the head of the machine
as soon as the head `points to' (i.e., enters) the current cell (see
Section~\ref{sec:TMencod}).

Ribbon extension is used to generate a sufficiently long
nesting of $cell$ ambients for the machine to run. A frozen ribbon
consists of a word $w$, containing at the end of the ribbon a frozen
ribbon extensor (definition of $\mactext{FrozenRibb}$ -- the
$\mactext{cleaninst}$ part will be useful later on).
The extensor is triggered by the presence of an ambient named $coin$
(definitions $\mactext{ExtensorFrozen}$ and
$\mactext{ExtensorAlive}$): when this happens, the loop programmed in
the definition of $\mactext{deadextcode}$ can start, which can have
the effect of adding new cells, whose value is $\ffalse$. Each time
the extensor loops (state $\mactext{ExtensorAlive}$), the $coin$
ambient can be erased by process $\openamb{coin}.\mactext{sendstart}$,
which has the effect of stopping the extension process, and sending an
ambient $msg$ out of the ribbon to instruct the machine to start
computation.  When this happens, the extensor is in
$\mactext{ExtensorDead}$ state.

A ribbon in $\mactext{GrowingRibb}$ state keeps extending until the
extensor dies, at which point it becomes a $\mactext{WorkRibb}$
($\mactext{WorkRibb}$ has two parameters, $w_1$ and $w_2$, in order
to reason about the cell where the head of the machine currently is).
Along this evolution, the $\mactext{cleaninst}$ code is always
present. When the machine successfully terminates computation (we will
describe below how this happens), it generates an ambient named
$cleaner$, which triggers the cleaning of the machine: all ambients
$cell, \ttrue, \ffalse, wo$, that intuitively constitute the ``data
structures'' of the machine, are removed. At this point, we obtain an
$\mactext{OldRibb}$.

\medskip

Some of the explanations we have just given are formalised by the
following result, which will be used to establish undecidability of
$\eqL$. 

\begin{lem}[Ribbon evolution]\label{fact_growing_ribbon}
~

  For any word $w$ and $n\in\N$, we write $P_n =
  \mactext{GrowingRibb}(w.(\ffalse)^n)$, where $(\ffalse)^n$ stands
  for the word written as $n$ times the name $\ffalse$. We have:
\begin{enumerate}[$\bullet$]
\item $P_n~\Rar~ P_{n+1}$; 
\item $P_n~\Rar~ R$
  with
  
  $R = \mactext{WorkRibb}(\epsilon,w.(\ffalse)^{n})\fillhole{
    \amb{msg}{!\outamb{cell}~|~
      \outamb{ribbon\_left}.\amb{start}{\inamb{TM}}}}$;
\item for any term $Q$ along the reduction paths from $P_n$ to
  $P_{n+1}$ and from $P_n$ to $R$, there exists $Q'$ such that
  $Q\,\equiv\,\amb{ribbon\_left}{Q'}$.
\end{enumerate}
Moreover, for any word $w$, we have:
\begin{mathpar}
  { \mactext{WorkRibb}(w,\epsilon)\fillhole{\nil}~|~
    \amb{cleaner}{\inamb{ribbon\_left}}~~ \Rar~~\mactext{OldRibb}\,.
    }
\end{mathpar}
  
\end{lem}

\begin{proof}
  At any step, the extensor can only choose between creating a new
  $\ffalse$ cell or dying and sending up through the ribbon an ambient
  $msg$. Note that when extending the ribbon with a new \ffalse{}
  cell, there are at some point two concurrent actions \inamb{cell}
  and \outamb{ext}: these are in causal dependency, since the
  \inamb{cell} can only happen once the \outamb{ext} has taken place,
  which ensures sequentiality of the execution.
\end{proof}

\subsection{Turing Machine}\label{sec:TMencod}

\newcommand{\rightmachine}{\ensuremath{\rightarrow}}

\begin{defi}[(Ideal) Turing Machine]
  We introduce three symbols $\leftarrow$, $\downarrow$ and
  \rightmachine{} for the movements of the head of a Turing Machine.
  
  We represent a Turing Machine as a quadruplet
  $(\QQ,q_{start},q_A,\delta)$ where $\QQ$ is a set of states,
  $q_{start}$ is the initial state, $q_A$ is the accepting state, and
  $\delta:\QQ\times\{\ffalse,\ttrue\}~\rar~
  \QQ\times\{\ffalse,\ttrue\}\times \{\leftarrow,\downarrow, \rightmachine\}$
  is the evolution function.
\end{defi}

\noindent\textbf{Notation:} we shall write
\begin{mathpar}
  {(w_1,q,w_2)~\transTuring~(w_1',q',w_2')}
\end{mathpar}
\noindent to denote the fact that the Turing Machine in state $q$
with the head on the cell of the last letter of $w_1$ (which will be
referred to as \emph{``the head dividing the ribbon into words $w_1$
  and $w_2$''}) evolves in one step of computation into the machine in
state $q'$, dividing the ribbon into words $w'_1$ and $w'_2$.

\medskip

The remainder of this subsection is devoted to establishing the
following claim:
\begin{claim}
Any Turing Machine computation may be encoded in \MAIFsyn.
\end{claim}

To encode Turing Machines, we must describe how we simulate in
\MAIFsyn{} the transitions of the machine, and how some extra
manipulations are performed after recognition of a word (these are
necessary to deduce the undecidability result proved below).

The encoding is given by the definitions collected in
Figure~\ref{fig:TM}. The overall shape of the encoding can be
described as follows:

\begin{defi}[Turing Machine in Mobile Ambients]\label{def:TM}
  The encoding of a Turing Machine is based on an ambient named $TM$,
  containing a persistent process named $\mactext{tmsoup}$:
  \begin{mathpar}
    { \mactext{tmsoup}~:=~~\mactext{code}(q_0)|\dots
      |\mactext{code}(q_n)~|~\mactext{getout}~|~!\openamb{mo}\,.}
\end{mathpar}
We define two configurations for the encoding of a Turing Machine.
Before being active, the machine is in \emph{starting state}, defined
by:
\begin{mathpar}
  { \mactext{TMStart}~:=~~\amb{TM}{~
      \openamb{start}.\inamb{ribbon\_left}.
      \inamb{cell}.\openamb{q_{start}}~|~\mactext{tmsoup}~}\,.  }
\end{mathpar}
Once the computation has started, the Turing Machine in state $q$ is
represented by the term
\begin{mathpar}
  {
    \mactext{TM}(q)~:=~~\amb{TM}{\,\openamb{q}~|~\mactext{tmsoup}\,}\,.}
\end{mathpar}
\end{defi}

\begin{figure}[t]
  $$\boxed{
  \begin{array}{lcl}
    \arraytopic{Turing Machine Transitions}\\
    \mactext{clear}(d);P & := & \amb{wo}{~\outamb{head}.\openamb{d}.
      \amb{cl\_ack}{\inamb{head}}~}~~|~~\openamb{cl\_ack}.P\\
    \mactext{write}(d);P & := & \amb{wo}{~\outamb{head}.\amb{d}{\nil}~|~
      \amb{wr\_ack}{\inamb{head}}~}~~|~~\openamb{wr\_ack}.P\\
    \mactext{become}(mo);P & := & 
    \amb{mo}{\outamb{head}.\openamb{head}.P}~|~ \inamb{mo}\\
    \mactext{domove}(mv);P & := & \left \{
      \begin{array}{ll}
        \inamb{cell}.P&\mbox{if}~mv=\leftarrow\\
        P&\mbox{if}~mv=\downarrow\\
        \outamb{cell}.P&\mbox{if}~mv=\rightmachine\\
      \end{array}\right.\\[.2em]
    \mactext{tcode}(d_r,q_w,d_w,mv) & := & \mactext{clear}(d_r);
    \mactext{write}(d_w);\\
    & & \qquad\qquad
    \mactext{become}(mo);\inamb{TM}.\mactext{domove}(mv);
    \openamb{q_w}\\[1em]
    \arraytopic{State}\\
    \multicolumn{3}{c}{
      \begin{array}{lcl}
    \ffalse\rar P~+~\ttrue\rar Q & := & 
    \amb{coin}{\inamb{\ffalse}.\outamb{\ffalse}.P}~|~
    \amb{coin}{\inamb{\ttrue}.\outamb{\ttrue}.Q}~|~\openamb{coin}\\
        \mactext{code}(q) & := & !q[head[~\outamb{TM}.\big(~~
        \ffalse\rar
        \mactext{tcode}(\ffalse,d_{\ffalse},q_{\ffalse},mv_{\ffalse})\\ 
        & &\qquad\qquad\qquad\qquad +~
        \ttrue\rar
        \mactext{tcode}(\ttrue,d_{\ttrue},q_{\ttrue},mv_{\ttrue})\big)~
        ]]\\ 
        & &|~!\amb{coin}{
          \inamb{\ffalse}.\outamb{\ffalse}.\mactext{tcode}(\ffalse, 
          d_{\ffalse},q_{\ffalse},mv_{\ffalse})}\\
        & &|~!\amb{coin}{
          \inamb{\ttrue}.\outamb{\ttrue}.\mactext{tcode}(\ttrue, 
          d_{\ttrue},q_{\ttrue},mv_{\ttrue})}\\[.2em]
        \mactext{code}(q_A) & := &  !\amb{q_A}{\,\amb{get\_out}{\nil}\,}
      \end{array}
      }\\[1em]
    \\
    \arraytopic{Turing Machine Behavior after Recognition}\\
    \multicolumn{3}{l}{
      \begin{array}{l}
        \mactext{getout}~:=
        \\
      \quad!\openamb{get\_out}.\outamb{cell}.\amb{get\_out}{\nil}\\ 
        |~!\openamb{get\_out}.\outamb{ribbon\_left}.
        \big(~\amb{cleaner}{\outamb{TM}.\inamb{ribbon\_left}}\\
        \hspace{5.5cm}|~~\amb{coin}{
          \outamb{TM}.\inamb{ribbon\_left}.\inamb{cell}^{\length{w}}. 
          \inamb{ext}}
        \\
        \hspace{5.5cm}
        |~~\openamb{start}.\inamb{ribbon\_left}.\inamb{cell}.\openamb{q_{start}}~~\big)
      \end{array}
      }\\
  \end{array}}$$
\begin{center}
  \caption{Encoding Turing Machines in \MAIFsyn}
  \label{fig:TM}
\end{center}
\end{figure}

\begin{lem}[$\Pbb$ encoding]
  All terms used in the encoding of a Turing Machine belong to $\Pbb$.
\end{lem}

Our Turing Machine encoding is somehow reminiscent of the one
presented in~\cite{CaGo98}. We should however remark that we work here
in a language without name restriction, and with a simpler encoding of
choice (operator $+$ above, to test the value of a cell).

According to the explanations given in Section~\ref{sec:ribbons}, the
machine reacts to the presence of an ambient named $start$ to enter
the first cell of the ribbon and start computation (definition
$\mactext{TMStart}$). 

The behaviour of the running machine is described by the definition of
$\mactext{code}(q)$: the head of the machine enters the current cell,
and tests its value by concurrently trying to enter ambients named
$\ffalse$ and $\ttrue$. According to the ambient being present, the
appropriate machine transition is triggered (definition of
$\mactext{tcode}$ --- $d_{\ffalse}, q_{\ffalse}, mv_{\ffalse}$
stand for the new value, new state, and movement of the head
determined by the current state if the value read is $\ffalse$, and
similarly for $\ttrue$). The last two lines in the definition of
$\mactext{code}$ (processes starting with $!coin\dots$) are there for
garbage collection purposes: they ``absorb'' the branch of the choice
that has not been triggered.

Performing a transition involves erasing the current value of the
cell, installing the new value, getting back inside the Turing
Machine (the current working ambient had to get out of it to read the
value of the cell), and triggering the movement of the machine
(definition of $\mactext{tcode}$). 
The corresponding definitions on top of Figure~\ref{fig:TM} should be
self-explanatory, the $\mactext{become}(mo)$ part being necessary to
synchronise with the $!\openamb{mo}$ inside ambient $TM$.
Finally, $\openamb{q_w}$ starts the execution of the code
corresponding to $q_w$, the new state of the machine --- according to
Definition~\ref{def:TM}, the code of all possible states of the
machine is present in replicated form in $TM$.

The code of the accepting state $q_A$ is peculiar: when the machine
reaches this state, it triggers process $\mactext{getout}$, which
makes it exit the ribbon and start the cleaning process. As explained
above, the presence of an ambient named $cleaner$ in ambient
$ribbon\_left$ triggers process $\mactext{cleaninst}$ of
Section~\ref{sec:ribbons}.
The process on the last line of Figure~\ref{fig:TM} is there to
install the machine in the \emph{exact} initial state once the word
has been recognized and cleaning has been performed.  This is
necessary to obtain a loop in the proof of Lemma~\ref{lemme_boucle}
below.

\medskip

We can remark that the encoding is parametric over a word $w$, whose
length (denoted $\length{w}$) is used in the definition of
$\mactext{getout}$ (in that definition, $\inamb{cell}^{\length{w}}$
stands for the concatenation of $\length{w}$ copies of the capability
$\inamb{cell}$). This aspect of our encoding is however irrelevant
since it is influent only after the end of the execution of the
machine, and not during the central part of the simulation.
%


\medskip

We now formulate the evolution of the terms we have defined in order
to simulate Turing Machines. We first introduce a useful relation.
\begin{defi}[deterministic evolution relation]
  We say that a process $P$ \emph{deterministically evolves} to $Q$,
  written $P~\cred~Q$, if and only if $P\rar Q$ and for any $Q'$ s.t.
  $P\,\rar\, Q'$, either $Q'\not\!\!\rar$ or $Q\,\equiv\, Q'$.
\end{defi}

\noindent\textbf{Notation:} We shall write $P\,\cred^k\, Q$ to say
that $P$ deterministically reduces to $Q$ in $k$ steps ($k\geq 1$).
We write $P\,\Cred\, Q$ when $P\,\cred^k\, Q$ for some $k$.

Using $\cred$, we can state some elementary facts about the macros
involved in the execution of the machine.
The relation $P\,\Cred\, Q$ captures the fact that $P$ cannot avoid
reducing to $Q$ except for some immediately blocking states. Such
blocking states may only appear due to the firing of the ``wrong
branch'' in a choice encoding ($\ffalse\rar\dots+\ttrue\rar\dots$).
(Incidentally, we may remark that a purely deterministic encoding of
the Turing Machine could probably be definable, but at the cost of
more complex definitions and proofs.)

\begin{lem}[state evolution]\label{macrosfact}
  For any terms $P,Q$, names $d, d'\in\{\ffalse,\ttrue\}$ and word
  $w$, we set $M~=~\amb{d}{\nil}~|~$
  $!\openamb{wo}~|~\mactext{word}(w)
  \fillhole{\mactext{ExtensorDead}}$. We then have the following
  deterministic transition sequences:
\begin{enumerate}[\em(1)]
\item $\amb{head}{d\rar P~+~\non d\rar
    Q}~|~\amb{d}{\nil}~|~!\openamb{wo}~|~
  \amb{cell}{M}~|~\amb{TM}{\mactext{tmsoup}}$
  $$
  \cred^3~~\amb{head}{P\,|\,\amb{coin}{\inamb{\non d}.\outamb{\non
        d}.Q}}|~\amb{d}{\nil}~|~!\openamb{wo}~|~
  \amb{cell}{M}~|~\amb{TM}{\mactext{tmsoup}}\,;
  $$

\item $\amb{head}{\mactext{clear}(d);P\,|\,\amb{coin}{\inamb{d'}.Q}}
  \,|\,\amb{d}{\nil}|\,|\,!\openamb{wo}\,|\,
  \amb{cell}{M}\,|\,\amb{TM}{\mactext{tmsoup}}$  
  $$
  \cred^5~~\amb{head}{P\,|\,\amb{coin}{\inamb{d'}.Q}}
  ~|~!\openamb{wo}~|~\amb{cell}{M}~|~\amb{TM}{\mactext{tmsoup}}\,;
  $$
\item $\amb{head}{\mactext{write}(d);P~|~\amb{coin}{\inamb{d'}.Q}}
  ~|~!\openamb{wo}~|~\amb{cell}{M}~|~\amb{TM}{\mactext{tmsoup}}$
  $$
  \cred^4~~\amb{head}{P\,|\,\amb{coin}{\inamb{d'}.Q}}
  ~|~\amb{d}{\nil}~|~!\openamb{wo}~|~\amb{cell}{M}
  ~|~\amb{TM}{\mactext{tmsoup}}\,;
  $$
\item $\amb{head}{\mactext{become}(mo);P\,|\,\amb{coin}{\inamb{d'}.Q}}
  \,|\,\amb{d}{\nil}\,|\,!\openamb{wo}\,|\,\amb{cell}{M}\,
  |\,\amb{TM}{\mactext{tmsoup}}$
  $$
  \cred^3~~\amb{mo}{P\,|\,\amb{coin}{\inamb{d'}.Q}}
  ~|~\amb{d}{\nil}~|~!\openamb{wo}~|~\amb{cell}{M}
  ~|~\amb{TM}{\mactext{tmsoup}}\,.
  $$
\end{enumerate}
Moreover, the same results hold with a frozen (instead of dead)
extensor in $M$, the only condition being that ambient $ext$ contains
an inactive term.
\end{lem}

\begin{proof}
  By inspection of the possible reductions of the processes being
  considered. From the second statement on, the ambient
  $\amb{coin}{\inamb{d'}.Q}$ is frozen: it actually represents the
  non-chosen branch in the encoding of the choice operator, that will
  be erased later, when the head of the Turing Machine comes back
  inside ambient $TM$ (see below).
\end{proof}

We can now merge the results above into a property regarding
transitions of the Turing Machine.

\begin{lem}[One step of Turing Machine simulation]
~

  Let $\mathcal{M}$ be a Turing Machine, $q$ one of its non accepting
  states, and $w_1,w_2$ two words, with $w_2\neq\epsilon$.  Suppose
  $(w_1,q,w_2)\transTuring(w_1',q',w_2')$. Then
  \begin{mathpar}
    { \mactext{WorkRibb}(w_1,w_2)\fillhole{\mactext{TM}(q)}~\Cred~
      \mactext{WorkRibb}(w_1',w_2')\fillhole{\mactext{TM}(q')}\,. }
\end{mathpar}
\end{lem}

\proof
  We divide the evolution of the term representing the Turing Machine
  into the following steps:
\begin{enumerate}[(1)]
\item From state $q$, the TM can trigger the $q$ code by performing
  the corresponding \openamb{\!} operation, which has the effect of
  releasing an ambient named $head$. Moreover, this is the only place
  where some reduction is possible, because first, \texttt{Extensor}
  is inactive and second, in every ambient named $cell$, no reduction
  occurs. Therefore,
$$
\begin{array}{l}
  \mactext{WorkRibb}(w_1,w_2)\fillhole{\mactext{TM}(q)}\\
  \quad\cred^2~
\mactext{WorkRibb}(w_1,w_2)\fillhole{\mactext{TMNostate}~|
~\amb{head}{\ffalse\rar\dots + \ttrue\rar\dots}}
\end{array}
$$
\noindent where the notation $\mactext{TMNostate}$ stands for the
following configuration of the Turing Machine ambient:
$$
\amb{TM}{\,\mactext{code}(q_0)~|\,\dots\,
  |~\mactext{code}(q_n)~|~\mactext{tmsoup}\,} 
$$
Note that this ambient cannot perform any reduction as long as it
is not visited by a $mo$ or $getout$ ambient.
\item Using the previous fact, and considering that reductions can
  only take place at $cell$ level, we have
$$
\begin{array}{l}
  \mactext{WorkRibb}(w_1,w_2)\fillhole{\mactext{TMNostate}~|
    ~\amb{head}{\ffalse\rar\dots +\ttrue\rar\dots }}~\cred^{15}~
  \\[.2em]
  \mactext{WorkRibb}(w_1^1\dots w_1^{r-1}d,w_2)\openhole{}
  \mactext{TMNostate}
  \\[.2em]
  \hspace{4cm}\,|\,
    \amb{mo}{\inamb{TM}.\mactext{domove}(mv).\openamb{q'}\,|\,
      \amb{coin}{\inamb{\non w_1^r}.P}}\closehole{}
\end{array}
$$
where $\delta(q,w_1^r)=(q',d,mv)$ (i.e., the machine evolves
from $q$ to $q'$ when reading $w_1^r$).
\item The ambient $mo$ comes back into the Turing Machine and is
  opened by the $\mactext{tmsoup}$ component.  Then the head movement
  (if any) is performed, which activates an $\openamb{q'}$ process, so
  that the Turing Machine gets into $\mactext{TM}(q')$ state.
  $$
  \begin{array}{r}
    \multicolumn{1}{l}{\mactext{WorkRibb}(w_1^1\dots w_1^{r-1}d,w_2)\openhole
      \mactext{TMNostate}}\\[.2em]
      \qquad\qquad\,|\,
  \amb{mo}{\inamb{TM}.\mactext{domove}(mv).\openamb{q'}\,|\,
    \amb{coin}{\inamb{\non w_1^r}.P}}\closehole
\\[.2em]
  ~\cred^{2(+1)}
  \mactext{WorkRibb}(w_1',w_2')\fillhole{\mactext{TM}(q')}\,.
  \end{array}
  $$
  Note that opening ambient $mo$ triggers the absorbtion of the
  non-selected branch of the choice (ambient $coin$) by a
  $!\amb{coin}{\dots}$ (from the code for the original state of the
  machine).
  
  The $2(+1)$ above comes from the fact that the head of the machine
  can also make no movement in its transition from a state to another
  (case $\downarrow$).\qed
\end{enumerate}\medskip

We obtain as a corollary of the Lemma above:
\begin{prop}[Turing Machine simulation]\label{prop_simulation}
  Given a Turing Machine $\mathcal{M}$, for any word $w$ and $n\in\N$,
  the Turing Machine $\mathcal{M}$ recognises the word $w$ on the
  ribbon $w.\ffalse^n$ iff there exist two words $w_1$ and $w_2$
  s.t.
  \begin{mathpar}
    {\mactext{WorkRibb}(\epsilon,
      w.\ffalse^n)\fillhole{\mactext{TM}(q_{start})} ~\Cred~
      \mactext{WorkRibb}(w_1,w_2)\fillhole{\mactext{TM}(q_A)}\,,}
\end{mathpar}
  \noindent where the terms above are given by the encoding of
  $\mathcal{M}$.
\end{prop}

Let us finally describe what happens after the machine has reached the
accepting state.

\begin{lem}[Acceptation]\label{fact_acceptation}
  Let $w_1, w_2$ be two words. Then
  \begin{mathpar}
    {
    \begin{array}{l}
      \mactext{WorkRibb}(w_1,w_2)\fillhole{\mactext{TM}(q_A)}\\
      \mbox{}\qquad
    \Rar~\mactext{OldRibb}~|~ \mactext{TMStart}~|~ \amb{coin}{
      \inamb{ribbon\_left}.\inamb{cell}^{\length{w}}.\inamb{ext}}
    \end{array}}
\end{mathpar}
  \noindent where $w$ is the word used in the encoding of the machine.
\end{lem}

\proof
We distinguish four steps:
\begin{enumerate}[(1)]
\item When the $q_A$ ambient has been opened, the ambient  $get\_out$
  is liberated and is present within $TM$:
$$
\mactext{WorkRibb}(w_1,w_2)\fillhole{\mactext{TM}(q_A)}~\Rar~
\mactext{WorkRibb}(w_1,w_2)\fillhole{\mactext{TMGetout}}
$$
where $\mactext{TMGetout}$ is the term
\begin{mathpar}
  {\amb{TM}{\amb{get\_out}{\nil}~|~\mactext{code}(q_0)~|\dots
      |~\mactext{code}(q_n)~|~\mactext{tmsoup}}
    \enspace.}
\end{mathpar}
\item This allows the $TM$ ambient to get a $get\_out$ `token',
  execute the branch containing the $\outamb{cell}$, and, doing this,
  liberate a new $get\_out$ ambient:
$$
\mactext{WorkRibb}(w_1,w_2)\fillhole{\mactext{TMGetout}}
~\Rar~
\mactext{WorkRibb}(w_1^1\dots w_1^{r-1},w_1^r.w_2)\fillhole{\mactext{TMGetout}}
$$
Note that the other subterm starting with $\openamb{get\_out}$
could also have been triggered, leading to a blocked state. This is no
harm for us, since we want to establish the existence of an execution
where the machine exits the ribbon.  This way, $TM$ progresses
outwards until it is directly inside $ribbon\_left$.
\item Then $TM$ gets out of $ribbon\_left$, choosing the other branch
  of $\openamb{get\_out}$, which leads to the following state:
$$
\begin{array}{cl}
\mactext{WorkRibb}(\epsilon,w_1.w_2)\fillhole{\nil}~~|~~ TM [ & \amb{cleaner}{\outamb{TM}.
\inamb{ribbon\_left}}\\
&|~\amb{coin}{\outamb{TM}.\inamb{ribbon\_left}.\inamb{cell}^{\length{w}}.
\inamb{ext}}\\
&|~\mactext{code}(q_0)\,|\dots|\,\mactext{code}(q_n)\,|\,\mactext{tmsoup}]
\end{array}
$$

\item At this point, the ambient named $TM$ may liberate an ambient
  $cleaner$ that enters $ribbon\_left$ and starts the cleaning
  process.  $TM$ may also liberate the ambient $coin$ so that we
  exactly obtain the expected term.\qed
\end{enumerate}

\begin{rems}\hfill
\begin{enumerate}[$\bullet$]
\item As we already mentioned above, our encoding of the Turing
  Machine is at this point dependent from the word $w$ that we want it
  to recognize.
\item reason here using $\Rar$ transitions instead of deterministic
  reduction $\cred$: indeed, we are considering states where the
  machine has already recognized the word, and we only need to prove
  that \emph{there exists} some way back to its (\emph{exact}) initial
  state. This will be enough for the proof of undecidability in
  Section~\ref{sec_undecid}.
\end{enumerate}
\end{rems}

\subsection{Undecidability of Logical Equivalence}\label{sec_undecid}

We can now exploit the encoding we have studied to establish
undecidability of \eqL.

\begin{lem}[Loop lemma]\label{lemme_boucle}
  Given a Turing Machine $\mathcal{M}$ and a word $w$, define the
  following terms, given from the encoding of $\mathcal{M}$:
  \begin{mathpar}
    {
  \begin{array}{c}
  Q ~:=~~!\mactext{FrozenRibb}(w)~|~!\mactext{OldRibb}~|~
  !\openamb{msg}~|~!\outamb{cell}~|~\mactext{TMStart}\,,\\
  P_0 ~:=~~ Q~|~\mactext{GrowingRibb}(w)\quad\mbox{and}\quad
  P_1 ~:=~~ Q~|~\mactext{GrowingRibb}(w.\ffalse)\,.
  \end{array}
  }
\end{mathpar}
  Then $P_0~\Rar~P_1$. Conversely, $P_1~\Rar~P_0$ if and only if the
  word $w$ may be recognized on a finite (but sufficiently long)
  ribbon of the shape $w.\ffalse^N$, for some $N\in\N$, by the
  Turing Machine $\mathcal{M}$.
\end{lem}

\begin{proof}
  The transition $P_0~\Rar~P_1$ follows from Lemma~\ref{fact_growing_ribbon}.
  
  Let us then first assume that $w$ can be recognized on a ribbon of
  the form $w.\ffalse^N$, that is, $w$ followed by an arbitrary number
  of \ffalse{} digits.  Then from Lemma~\ref{fact_growing_ribbon}, we
  can obtain the corresponding extension of the ribbon from state
  $P_1$, i.e.  exhibit a transition $P_1~\Rar~
  Q~|~\mactext{WorkRibb}(w.\ffalse^N,\epsilon)\fillhole{\nil}~|~
  \amb{start}{\inamb{TM}}$.  At this point, the ambient $start$ can
  enter $TM$ and allow it to get into the work ribbon. Then, using the
  simulation result (Proposition~\ref{prop_simulation}), we know that
  the Turing Machine reaches the acceptation state (this result is
  obtained by induction over the length of $w$). At this point,
  according to Lemma~\ref{fact_acceptation}, the work ribbon is
  transformed into an old ribbon (collected by the corresponding
  replicated term in $Q$), the Turing Machine comes out of the ribbon,
  and waits for a start signal. The liberated $coin$ ambient may
  progress inside a frozen ribbon (containing word $w$ by definition
  of $Q$ above) until it reaches the frozen extensor and wakes it up.
  We then exactly obtain $P_0$.
  
  Now let us assume that $w$ cannot be recognized on any ribbon. As
  $Q$ is blocked (in particular, \texttt{TMStart} is waiting for an
  ambient $start$ to enter $TM$), the first reducts of $P_1$ are of
  the form $Q~|~\amb{ribbon\_left}{R}$, where
  $\mactext{GrowingRibb}(w.\ffalse) ~\Rar ~\amb{ribbon\_left}{R}$.
  If a reduction chain from $P_1$ to $P_0$ can be found, then by
  Lemma~\ref{fact_growing_ribbon} there exists an integer $n$ such that
  $$P_1~~\Rar~
  \underbrace{Q~|~\mactext{WorkRibb}(w.\ffalse^n)\fillhole{\nil}~|
    ~\amb{start}{\inamb{TM}}}_T~~\Rar~ P_0\,.$$
  In term $T$ the \texttt{WorkRibb} is blocked, so the only
  evolution can come from the machine entering a ribbon. We
  distinguish three cases according to the kind of ribbon which is
  entered by the machine:
  \begin{enumerate}[(1)]
  \item If it gets into an old ribbon, there can be no more reduction,
    as the $TM$ is stuck on an $\inamb{cell}$ action.  
  \item If it gets into the work ribbon, according to
    Proposition~\ref{prop_simulation}, there is a unique way to
    evolve, through simulation of the machine.  At this point, the
    machine may have an infinite computation on the finite ribbon,
    never reaching accepting state: this means that it will not get
    out of the ribbon, which prevents the system to evolve into $P_0$.
    Alternatively, the machine may try to use more ribbon than what
    has been created before evolution from \texttt{GrowingRibb} into
    \texttt{WorkRibb}, and the machine is stuck. So in any case,
    state $P_0$ cannot be reached.
  \item We reason similarly in the case where the machine enters a
    frozen ribbon.
  \end{enumerate}
  Finally, we have that state $P_0$ is unreachable if word $w$ cannot
  be recognised by the machine on a ribbon of the form $w.\ffalse^N$
  for some $N$, which concludes the proof.
\end{proof}

\begin{thm}[Undecidability of $\eqL$]\label{theorem_undecidability}
  $\eqL$ is an undecidable relation on \MA.
\end{thm}
\begin{proof}
  Let us first note that the decidability of $\eqL$ over $\MAIF$ is a
  consequence of its inductive characterisation $\bigind$
  (Definition~\ref{d:bigind}) together with the image finitess
  hypothesis of $\MAIF$.
  
  Consider processes $P_0$ and $P_1$ from Lemma~\ref{lemme_boucle}.
  We show that the problem of deciding whether one can prove
  $\openamb{n}.P_0~\eqL~\openamb{n}.P_1$ is equivalent to deciding
  whether $P_0\Rar P_1\Rar P_0$. This will be enough, by
  Lemma~\ref{lemme_boucle}, to obtain the undecidability of $\eqL$.
  
  


  Let us prove now the undecidability of $\eqL$ on $\MA$.  Consider
  processes $P_0$ and $P_1 $ of Lemma~\ref{lemme_boucle}.  These
  processes are in $\Pbb$. Using Corollary~\ref{c:chara}, the
  definition of $\bisMOD$, and Theorem~\ref{t:bisSTR}, we have:
  \begin{mathpar}
    {
  \begin{array}{rcl}
\openamb{n}.P_0~\eqL~\openamb{n}.P_1 & \mbox{iff}&
\openamb{n}.P_0~\bisMOD~\openamb{n}.P_1\\
 &\mbox{iff} &
P_0~\Rar\bisMOD~P_1~\Rar\bisMOD~P_0
  \end{array}}
\end{mathpar}
(from Theorem~\ref{t:bisSTR}, $\Rar\bisMOD$ is $\Rar$ on $\Pbb$).

The first equivalence follows from soundness and completeness
(Theorems~\ref{t:soundnessLogBis} and \ref{t:completeness}). The
second is the definition of $\bisMOD$. Since on $\Pbb$ $\bisMOD =
\equiv$, the last condition is simply the loop condition, and
undecidability follows from Lemma~\ref{lemme_boucle}.
\end{proof}



\section{Conclusions and future work}\label{s:concl}

In this paper we have presented a number of 
 characterisations of logical equivalence, including a
coinductive characterisation by means of intensional bisimilarity,
$\bisMOD$, and 
an inductive characterisation based on inversion results
for $\bisMOD$.
These characterisation results
are established  on the MA calculus in which
 terms need not be image-finite, and with respect to a
finitary  logic.   We are not aware of other results of this kind.
(Characterisation results for a bisimilarity with respect to a modal
logic in the literature  rely either 
on an  image-finiteness
hypothesis for the terms of the language, or  on the presence of some
infinitary constructs in the
syntax of the logic.)

We have compared logical equivalence with
 barbed congruence, showing that the latter is 
strictly coarser, and  with
structural congruence, showing that the two relations
are ``almost the same''
in the (Turing-complete) calculus \MAIFsyn 
 (the two relations coincide on the synchronous version of
\MAIFsyn, whereas an additional eta-law has to be added to
structural congruence in the asynchronous calculus).
A spin-off of this study is a general
  better understanding of behavioural
equivalences in Ambient-like calculi. 
For instance,   we have shown that 
behavioural  equivalences can be
 insensitive to
{stuttering} phenomena originated by processes that may
repeatedly   enter and exit an ambient. 
Finally, we have proved
that logical equivalence, although decidable on \MAIFsyn, it  is not decidable on the whole MA calculus.



\medskip

We discuss below a few possible extensions of our work. 
On the logic side, 
 other logical connectives 
could be added without changing
our results,  as long as formulas expressing capabilities and
replication can still be derived. 
We believe this holds in particular  for  
the  `somewhere'
modality~\cite{Cardelli::Gordon::AnytimeAnywhere}, and for fresh
quantification~\cite{gabbay:pitts:fac02}.

In our work,
 we have
interpreted the `sometimes' modality ($\sometime \AAA$) in a weak sense,
which makes intensional bisimilarity a weak form of  bisimilarity.  We believe
that under a strong interpretation of the modality the result
corresponding to Theorem~\ref{t:bisSTR} can be derived in a much
simpler way, especially  because stuttering does not show up.

On the calculus side, 
 a first variation could be the introduction
 of  a general recursion
scheme instead of replication. This would make it possible to express
recursion `in depth', and not only `in width', as with replication.
Our proofs do not obviously carry over  to this setting, mainly due to the
fact that the sequential degree of a process may then be infinite, and
we would lack a measure to reason by induction.

Another interesting extension is the  addition of  name restriction $(\nu n)P$ to
the calculus.  Including restriction naturally implies to add its
logical counterpart, name revelation ($n\reveal \AAA$,
see~\cite{Cardelli::Gordon::NameRestriction::01}) to the logic. Our
results can be extended to this setting on the finite calculus, and on
infinite processes with only finitely many restricted names, but we do
not know how to extend them to richer calculi. For instance, the proof
of completeness cannot be directly adapted to the extension with name
restriction in the general case.
The possibility of generating infinitely many fresh names breaks
Lemma~\ref{l:finite}, intuitively because infinitely many frozen
subterms can appear as outcomes of a given term.
For the same reason, we think that our approach to obtain completeness
in absence of an image-finiteness hypothesis cannot be adapted to the
$\pi$-calculus, where infinitely many names can be generated.
However, our results for the \MAIFsyn{} fragment, in particular
Theorem~\ref{t:bisSTR} ($\eqL=\,\equiv$), still hold in presence of
name restriction.

In the paper we have considered only communications of basic names.
Certain presentation of the MA calculus also include operators for
communication  of capabilities. 
We believe that such communications 
could be added with mild modifications to the proofs.

\end{document}